\documentclass[11pt]{article}
\usepackage{float} %Ojo, hay que cargarlo antes de hyperref que va con marsden_article, sino da problemas
\usepackage{marticle}
\usepackage[all]{xy}
\usepackage{amssymb}
\usepackage{enumerate,subcaption}
\usepackage{tikz-cd}
\usetikzlibrary{decorations.pathmorphing}
\usepackage[siunitx]{circuitikz}
\usepackage{verbatim}
\usepackage{xcolor}
\usepackage{bold-extra}

\graphicspath{ {images/} }

\newtheoremstyle{obs}% name
  {3pt}%      Space above
  {3pt}%      Space below
  {}%         Body font
  {}%         Indent amount (empty = no indent, \parindent = para indent)
  {\bfseries}% Thm head font
  {.}%        Punctuation after thm head
  {.5em}%     Space after thm head: " " = normal interword space;
  {}%         Thm head spec (can be left empty, meaning `normal')

\theoremstyle{obs}

\newtheorem{remark}[theorem]{Remark}

\newcommand{\R}{\mathbb{R}}

\newcommand{\lcf}{\lbrack\! \lbrack}
\newcommand{\rcf}{\rbrack\! \rbrack}

\newcommand{\lvec}[1]{\overleftarrow{#1}}
\newcommand{\rvec}[1]{\overrightarrow{#1}}

\newcommand{\st}{\;\ifnum\currentgrouptype=16 \middle\fi|\;}

\makeindex

\usepackage[backend=bibtex, style=numeric-comp, isbn=false, doi=false, url=false]{biblatex}
\begin{document}
\title{Variational order for forced Lagrangian systems II:  Euler-Poincar\'e equations with forcing }

\author{ D. Mart\'{\i}n de Diego,  R. T. Sato Mart{\'\i}n de Almagro
\\[2mm]
{\small   Instituto de Ciencias Matem\'aticas (CSIC-UAM-UC3M-UCM)} \\
{\small C/Nicol\'as Cabrera 13-15, 28049 Madrid, Spain}}

%\\[2mm]\textsection Department of Mathematics, Faculty of Science, The University of Ostrava,\\ 30. dubna 22, 701 03 Ostrava, Czech Republic
%\\[2mm]$\dagger$ Department of Mathematics, Ghent University,\\ Krijgslaan 281, B--9000 Gent, Belgium
%\\[2mm] $\ddagger$ La Laguna}

\date{\today}

\maketitle

\begin{abstract}
In this paper we provide a variational derivation of the Euler-Poincar\'e equations for systems subjected to external forces using an adaptation of the techniques introduced by Galley and others \cite{Galley13, Galley14, Sato}. Moreover, we study in detail the underlying geometry which is related to the notion of Poisson groupoid. Finally, we apply the previous construction to the formal derivation of the variational error for numerical integrators of forced Euler-Poincar\'e equations and the application of this theory to the derivation of geometric integrators for forced systems.
\end{abstract}

\tableofcontents

\section{Introduction}
In the last few decades, there has been a steady research effort in the area of discrete variational mechanics (see \cite{MR1666871,marsden-west,hairer} and references therein). This interest stems from its application to the generation of numerical methods that can be derived from it and their desirable properties. The solution of a discrete variational problem directly inherits characteristic structural properties from its continuous counterpart such as, for instance, symplecticity, momentum or configuration space preservation, or an excellent energy behaviour...
Another reason for using discrete variational techniques is that this kind of methods typically admits an easy adaptation to other important dynamical systems such as, for instance, forced Lagrangian systems, nonholonomic dynamics, systems reduced by a Lie group of symmetries, non-smooth frameworks or even classical field theories (see \cite{boma,CM2001,GNI,leoshin,MMM06Grupoides,moser-veselov,OJM,perlmutter06}) among many others). 

An important result for the construction of methods based on discrete variational calculus is the variational error analysis theorem \cite{marsden-west} since it considerably lowers the difficulty of proving the order of a proposed integrator. Marsden and West considered this construction for Lagrangians defined on the tangent bundle $TQ$ of the configuration space $Q$ and the local error analysis of variational integrators defined on the corresponding  discrete space $Q \times Q$. This result was later completely and rigorously stablished in \cite{PatrickCuell}. In our recent paper, we also gave a rigorous proof of the result for Lagrangian systems subjected to external forces (see \cite{Sato}).

In this paper, we restrict our attention to forced Lagrangian systems whose configuration space is a Lie group which additionally admit reduction by left or right action, i.e., Euler-Poincar\'e equations with forcing (see \cite{marsden3,Holm1,Holm2}). As a paradigmactic example, we can think of a rigid body subjected to external forces depending only on its angular velocity. We will also prove how to derive high-order methods for Euler-Poincar\'e equations with external forces using standard variational error analysis. 

We will adapt the duplication of variables technique that was used previously in \cite{Sato} to the case of Euler-Poincar\'e systems. As we will see, the extension of our result to this case is far from trivial and we will need to use the geometric notion of Poisson groupoid. This notion was introduced in \cite{weipoisson} in order to unify the theories of Poisson Lie groups and symplectic groupoids (see also \cite{ping}).

Finally, we will show the performance of the geometric integrators that we construct in concrete examples, specifically in cases where the forced systems also have some important preservation properties such as, for instance, cases when the energy is dissipated but angular momentum is not. In more geometric terms, this means that the coadjoint orbits remain invariant, but on them the energy is decreasing along orbits (see \cite{BKMR}). 

\section{Euler-Poincar\'e and Lie-Poisson equations with forcing}

\subsection{Euler-Poincar\'e equations with forcing}
Let $G$ be a finite dimensional Lie group as the configuration space of a mechanical system \cite{AM87,marsden3,Holm1,Holm2}. Using the left (or right, alternatively) multiplication ${\mathcal L}_g: G\rightarrow G$, ${\mathcal L}_g(g')=gg'$ allows us to trivialize the tangent bundle $TG$ and the cotangent bundle $T^*G$ as follows
\begin{alignat*}{4}
TG &\to G\times \mathfrak{g}\phantom{^*}, \qquad &(g, \dot{g}) &\mapsto (g, T_g{\mathcal L}_{g^{-1}}\dot g)=(g, \eta)\; ,\\
T^*G &\to G\times \mathfrak{g}^*, \qquad &(g, p) &\mapsto (g, T^*_e {\mathcal L}_g p) = (g, \mu)\; ,
\end{alignat*}
where $\mathfrak{g}=T_eG$ is the Lie algebra of $G$ and $e$ denotes the identity element of $G$. In this paper we will work preferably with the left translation but it is possible to derive equivalent results using right translation. Throughout the paper we will also use matrix notation, i.e. $T_g{\mathcal L}_{g^{-1}}\dot g \equiv g^{-1}\dot{g}$ intermittently to simplify computations.

Given a Lagrangian $L: TG \to \mathbb{R}$, we define its left-trivialized version, $\check{L}: G \times \mathfrak{g} \to \mathbb{R}$, by the relation
\begin{equation*}
\check{L}(g,\eta) = L(g, T_e {\mathcal L}_{g}\eta)\,.
\end{equation*}

With this trivialized Lagrangian, the classical Euler-Lagrange equations can be rewritten as
\begin{align}
\frac{\mathrm{d}}{\mathrm{d}t}\left(\frac{ \partial \check{L}}{\partial \eta}\right) &= \mathrm{ad}_{\eta}^*\frac{\partial \check{L}}{\partial\eta}+T^*_e{\mathcal L}_g\left(\frac{\partial \check{L}}{\partial g}\right)\label{elt}\\
\dot{g} &= g \eta = T_e{\mathcal L}_g \eta
\end{align}
where $\mathrm{ad}_{\xi}\eta = [\xi, \eta]$, with $\xi, \eta \in \mathfrak{g}$, and $[\; , \;]$ denotes the Lie bracket on the algebra.

The Euler-Lagrange equations are modified under the presence of external forces, which are introduced as a map $F: TG\rightarrow T^*G$. The resulting equations can be obtained by applying the Lagrange-D'Alembert principle and under left-trivialization they become
\begin{equation}
\frac{\mathrm{d}}{\mathrm{d}t}\left(\frac{ \partial \check{L}}{\partial \eta}\right) = \mathrm{ad}_{\eta}^*\frac{\partial \check{L}}{\partial\eta}+T^*_e{\mathcal L}_g\left(\frac{\partial \check{L}}{\partial g}\right) + \check{F}.
\label{eq:elf}
\end{equation}
where $\check{F}: G \times \mathfrak{g} \to \mathfrak{g}^*$ is defined by
\begin{equation*}
\check{F}(g, \eta) = T^*_e {\mathcal L}_g \left(F(g, T_e {\mathcal L}_{g}\eta)\right).
\end{equation*}

In the sequel we will assume that $L$ is left invariant, i.e. the following reduced Lagrangian, $l: \mathfrak{g}\rightarrow \mathbb{R}$ 
\begin{equation*}
l(\eta) = L(e, \eta) = L(g'g, g'\dot{g}), \forall g' \in G
\end{equation*}
is well defined. That is, $l$ is the restriction of $L$ to $\mathfrak{g}$. In this case, the corresponding free Euler-Lagrange equations are:
 \begin{equation}\label{epe}
 \frac{\mathrm{d}}{\mathrm{d}t}\left(\frac{ \partial l}{\partial \eta}\right) = \mathrm{ad}_{\eta}^*\frac{\partial l}{\partial\eta}\,.
\end{equation}
These are known as the (left-invariant) {\bf Euler-Poincar\'e equations}.

The procedure to find a solution $t \mapsto g(t)$ of the Euler-Lagrange equations with initial condition $g(0) = g_0$ and $\dot{g}(0) = v_0$ is the following. First, we solve the first order ODE system defined by the Euler-Poincar\'e equations \eqref{epe} with the initial condition $\eta(0) = g_0^{-1} v_0$. With this solution, $t \mapsto \eta(t)$, we solve the so-called reconstruction equation: 
\begin{equation*}
\dot{g}(t) = g(t)\eta(t), \quad \text{with } g(0)=g_0\, .
\end{equation*}

The Euler-Poincar\'e equations are modified under the presence of external forces that, mathematically, are introduced as a map $f: \mathfrak{g}\rightarrow \mathfrak{g}^*$. The Euler-Poincar\'e equations with forcing are  
\[
 \frac{\mathrm{d}}{\mathrm{d}t}\left(\frac{ \partial l}{\partial \eta}\right)=\mathrm{ad}_{\eta}^*\frac{\partial l}{\partial\eta}+f(\eta).
\]
A relation between this equation and eq.\eqref{eq:elf} can be readily stablished if $L$ is left-invariant and $F$ does not depend on the point of application, as then
\begin{equation*}
f(\eta) = T^*_e{\mathcal L}_g {F}(e,\eta)
\end{equation*}
By fixing a basis $\{e_a\}$ of the Lie algebra $\mathfrak{g}$ we may induce coordinates $(\eta^a)$ so that we can write $\eta = \eta^a e_a$. Then, the Euler-Poincar\'e equations with forcing have the following expression in local coordinates
\[
 \frac{\mathrm{d}}{\mathrm{d}t}\left(\frac{ \partial l}{\partial \eta^a}\right)=C^d_{ba}\eta^b \frac{\partial l}{\partial \eta^d} +f_a
\]
where $C_{ab}^c$ are the structure constants of the Lie algebra $\mathfrak{g}$, that is, $[e_a, e_b]=C_{ab}^ce_c$ and $\langle f(\eta), e_a\rangle=f_a$.

\subsection{Lie-Poisson equations with forcing}
(See \cite{david,BKMR} for a longer and deeper exposition.)

Given a Hamiltonian $h: \mathfrak{g}^*\rightarrow \mathbb{R}$ then the Lie-Poisson equations are 
\begin{align}
	\dot{\mu} &= \mathrm{ad}_{h'(\mu)}^*\mu,\label{lpe}\\
	\dot{g} &= g h'(\mu)\label{elth},
\end{align}
where $\mu\in \mathfrak{g}^*$ and $h'(\mu)=\partial h/\partial \mu$.

Given a regular Lagrangian $l: \mathfrak{g}\rightarrow \mathbb{R}$, the Legendre transformation
\begin{equation*}
	\begin{array}{rrcl}
	{\mathcal F}l: & \mathfrak{g} & \longrightarrow & \mathfrak{g}^*\\
	               &          \eta &     \longmapsto & \displaystyle \mu = \frac{\partial l}{\partial \eta}
	\end{array}
\end{equation*}
is a local diffeomorphism and we can relate equations \eqref{epe} and \eqref{lpe}. If we define the Lagrangian \emph{energy function},
\begin{equation*}
	E_l(\eta) = \left\langle \frac{\partial l}{\partial \eta}, \eta \right\rangle - l(\eta),
\end{equation*}
then this implicitly defines a Hamiltonian function $h \circ \mathcal{F}l = E_l$.
%then using that there exists a unique $\eta$ such that $\mu = l'(\eta)$, for each $\mu \in \mathfrak{g}^*$, this in turn defines a Hamiltonian function
%\begin{equation*}
%H(\mu) = \langle \mu, \eta (\mu)\rangle - l(\eta(\mu)).
%\end{equation*}

%It is well known that both brackets are induced by reduction of the standard Lie bracket on $T^*G$ by right or left-%translation.
It is important  to note that $\mathfrak{g}^*$  is equipped with the Lie-Poisson bracket $\{\; ,\; \}$ 
\begin{equation}
	\{f, g\} = -\left\langle \mu, \left[ \frac{\partial f}{\partial \mu}, \frac{\partial g}{\partial \mu}\right]\right\rangle
\end{equation}
where $f, g\in C^{\infty}(\mathfrak{g}^*)$. 

In coordinates $\mu_a$, induced by the dual basis $\{e^a\}$ on $\mathfrak{g}^*$, we have that 
\begin{equation*}
	\{\mu_a, \mu_b\} = -C_{ab}^d \mu_d\,.
\end{equation*}

This bracket   exactly corresponds to the reduced bracket by standard  Poisson reduction from
$
	\pi: (T^*G, \omega_G) \longrightarrow (T^*G/G \equiv \mathfrak{g}^*, \{\; , \; \})
$
where $\pi(\mu_g) = [\mu_g] \equiv T_e^*{\mathcal L}_g(\mu_g)$. 

Fixed  $\mu_0\in \mathfrak{g}^*$, the coadjoint orbit is
$
	{\mathcal O}_{\mu_0} := \left\{ \mathrm{Ad}^*_{g^{-1}}\mu_0\; \vert \; g\in G\right\}\subseteq \mathfrak{g}^*.
$
If $t\rightarrow \mu(t)$ is the solution of  the initial value problem $\dot{\mu}=\mathrm{ad}_{h'(\mu)}^*\mu$ with $\mu(0)=\mu_0$ then  we can deduce that $\mu(t)\in {\mathcal O}_{\mu(0)}$.

Given a Hamiltonian function $h: \mathfrak{g}^*\rightarrow \mathbb{R}$ we derive the equations of motion by
the equations
\begin{equation}\label{plo}
\dot{\mu}(t)=\sharp^{\Pi}(\mathrm{d} h(\mu(t)))
\end{equation}
where $\Pi$ is the bivector field associated to the bracket $\{\; ,\; \}$. 
It is well known that the flow $\Psi^h_t: \mathfrak{g}^*\rightarrow \mathfrak{g}^*$ of $X_h$ verifies some geometric properties: 
\begin{enumerate}
	\item It preserves the linear Poisson bracket, that is
	\begin{equation*}
		\left\lbrace f\circ \Psi_t, g\circ \Psi_t\right\rbrace = \{f, g\} \circ \Psi_t, \quad \forall f, g \in C^{\infty}(\mathfrak{g}^*).
	\end{equation*}
	\item It preserves the hamiltonian 	$h \circ \Psi^h_t = h.$
	\item If all the coadjoint orbits are connected, Casimir functions are also preserved along each coadjoint orbit. 
\end{enumerate}

We can also add forces in our picture, in this case modeled as a map $\tilde{f}: \mathfrak{g}^*\rightarrow \mathfrak{g}^*$. 
If we start from a force on the Lagrangian side $f: \mathfrak{g}\rightarrow \mathfrak{g}^*$ then we define the force on the Hamiltonian side by taking $\tilde{f} \circ \mathcal{F}l = f$. The Lie-Poisson equations with forcing are modified as folllows: 
  \begin{equation*}
	  \dot{\mu} = \mathrm{ad}_{h'(\mu)}^*\mu + \tilde{f}(\mu).
  \end{equation*}

It is clear that by adding the forcing term we are losing all the properties of the flow of the free system (preservation of the Hamiltonian, preservation of coadjoint orbits...). But in the case of forces $\tilde{f}: \mathfrak{g}^*\rightarrow \mathfrak{g}^*$ of the special form
\begin{equation*}
	\tilde{f}(\mu) = \mathrm{ad}^*_{\tilde{\zeta}(\mu)}\mu
\end{equation*}
where $\tilde{\zeta}: \mathfrak{g}^* \to \mathfrak{g}^*$ is an arbitrary map, in this particular case, the coadjoint orbits are preserved. Observe that in this last case the Euler-Poincar\'e equations  are transformed like
\begin{equation*}
	\frac{\mathrm{d}}{\mathrm{d}t}\left(\frac{ \partial l}{\partial \eta}\right)=\mathrm{ad}_{\eta+\zeta(\eta)}^*\left(\frac{\partial l}{\partial\eta}\right)
\end{equation*}
for an arbitray map $\zeta: \mathfrak{g}\rightarrow \mathfrak{g}$ (see \cite{BKMR}).

\subsection{Discrete Lagrangian formalism}
Now, we will describe discrete Euler-Poincar\'e equations 
(see \cite{MR1726670,MR2496560,MMM06Grupoides} for more details). 
Fixed an element $W\in G$, define the set of admissible pairs
\begin{equation*}
	C_W^2 = \left\lbrace (W_1, W_2) \in G \times G \,\vert\, W_1 W_2 = W \right\rbrace.
\end{equation*}

A tangent vector to the manifold $C_W^2$ is a tangent vector at $t=0$ of a curve in $C_W^2$
\begin{equation*}
	t\in (-\epsilon, \epsilon)\subseteq \R\longrightarrow (c_1(t), c_2(t))
\end{equation*}
where $c_i(t)\in G$, $c_1(t)c_2(t) = W$ and $c_1(0) = W_1$ and $c_2(0) = W_2$.
These types of curves  are given by
\begin{equation}\label{curves}
	c(t) = (W_1 U(t), U^{-1}(t) W_2)
\end{equation}
for an arbitrary $U(t)\in G$ with $t\in (-\epsilon, \epsilon)$ and $U(0)=e$, where $e$ is the identity element of $G$.

Fixed a discrete Lagrangian $l_d: G\rightarrow \R$, we define the {\bf discrete action sum} by
\begin{equation*}
	\begin{array}{rrcl}
		S_{l_d}: &     C^2_W  & \longrightarrow & \mathbb{R}\\
		         & (W_1, W_2) &     \longmapsto & l_d(W_1) + l_d(W_2)
	\end{array}
\end{equation*}

\begin{definition}{\bf [Discrete Hamilton's principle]}
Given $W \in G$, then $(W_1, W_2)\in C_W^2$ is a solution of the discrete Lagrangian system determined by $l_d: G\rightarrow \R$ if and only if $(W_1, W_2)$ is a critical point of $S_{l_d}$.
\end{definition}

We characterize the critical points  using the curves defined in (\ref{curves}) as follows
\begin{align*}
	0 &= \left.\frac{\mathrm{d}}{\mathrm{d}t} \right\vert_{t=0} S_{l_d}(c(t))\\
	  &= \left.\frac{\mathrm{d}}{\mathrm{d}t} \right\vert_{t=0} \left[ l_d(W_1U(t)) + l_d(U(t)^{-1}W_2)\right]\\
	  &= \left\langle \mathcal{L}^{*}_{W_1} l_d'(W_1) - \mathcal{R}^{*}_{W_2} l_d'(W_2) ,\zeta\right\rangle
\end{align*}
%\begin{align*}
%	0 &= \left.\frac{\mathrm{d}}{\mathrm{d}t} \right\vert_{t=0} S_{l_d}(c(t))\\
%	  &= \left.\frac{\mathrm{d}}{\mathrm{d}t} \right\vert_{t=0} \left( l_d(W_1U(t)) + l_d(U(t)^{-1}W_2)\right)\\
%	  &= \mathrm{d}(l_d\circ {\mathcal L}_{W_1})(e)(\xi)- \mathrm{d}(l_d\circ {\mathcal R}_{W_2})(e)(\xi)
%\end{align*}
where $\zeta = \dot{U}(0)$. 

Alternatively, we can write these equations as follows
\begin{equation}\label{dep}
	0 = \lvec{\xi}_{W_1} (l_d)- \rvec{\xi}_{W_2} (l_d)\; , \quad \forall \xi \in \mathfrak{g}
\end{equation}
which are called {\bf discrete Euler-Poincar\'e equations}. Here $\lvec{\xi}_U = T_e \mathcal{L}_U \xi$ and $\rvec{\xi}_U = T_e \mathcal{R}_U \xi$ are the left and right-invariant vector fields, respectively. 

Also it is possible to define two discrete Legendre transformations by 
$\mathcal{F}^- l_d:  G \rightarrow \mathfrak{g}^*$ and $\mathcal{F}^+ l_d:  G \rightarrow \mathfrak{g}^*$ by
\begin{align*}
	\mathcal{F}^- l_d(W) &= \mathcal{R}_W^* l'_d(W)\\
	\mathcal{F}^+ l_d(W) &= \mathcal{L}_W^* l'_d(W)
\end{align*}
So, if we define
\begin{equation*}
	\mu_k = \mathcal{F}^- l_d(W_k) = \mathcal{R}_{W_k}^* l'_d
\end{equation*}
then equation \eqref{dep} are equivalent to
\begin{equation*}
	\mu_2 = \mathcal{F}^- l_d(W_2) = \mathcal{F}^+ l_d(W_1) = \mathrm{Ad}^*_{W_1} \mu_1
\end{equation*}
which in this case are called {\bf discrete Lie-Poisson equations}. Then, an implicit map $\mu_k \mapsto \mu_{k+1}$ is defined such that it preserves the Lie-Poisson structure. If the discrete Lagrangian function $l_d: G\rightarrow \mathbb{R}$ is regular, that is, the Legendre transformation $\mathcal{F}^- l_d:  G\rightarrow \mathfrak{g}^*$ is a local diffeomorphism (or, equivalently, $\mathcal{F}^+ l_d: G \rightarrow \mathfrak{g}^*$ is a local diffeomorphism) then the map $\mu_k \mapsto \mu_{k+1}$ is, in fact, explicit.

To obtain a numerical integrator for the dynamics determined by a continuous Lagrangian $l: \mathfrak{g}\rightarrow {\mathbb  R}$ it is necessary to know  how closely the trajectory of the proposed  numerical method  matches the exact trajectory of the Euler-Poincar\'e equations. For variational integrators, an important tool for simplifying the error analysis is to alternatively study how closely a discrete Lagrangian matches the exact discrete Lagrangian defined by $l: \mathfrak{g}\rightarrow {\mathbb  R}$. In our case, the exact Lagrangian is given by 
\begin{equation}
	l_h^{e}(W) = \int_{0}^{h} l(\eta_W(t)) \mathrm{d}t, \text{ for } g \in G,
\end{equation}
where $\eta_W: I \subseteq \mathbb{R} \to \mathfrak{g}$ is the unique solution of the Euler-Poincar\'e equations for $l: \mathfrak{g} \to \mathbb{R}$ such that the corresponding solution $(g, \dot{g}): I \subseteq \mathbb{R} \to TG$ of the Euler-Lagrange equations for $L(g, \dot{g}) = l\left( g^{-1}g \right)$ satisfies
\begin{equation*}
	g(0) = e, \quad g(h) = W.
\end{equation*}
In \cite{MMM3} it is shown that if we take as a discrete Lagrangian an approximation of order $r$ of the exact discrete Lagrangian, then, the associated  discrete evolution operator is also of order $r$, that is, the derived  discrete scheme is an approximation of  the continuous flow of order $r$. 

%\textcolor{red}{Add systems with force} 

\section{A variational description of forced Euler-Poincar\'e equations}

In this section, we will study a purely variational description of the Euler-Poincar\'e and Lie-Poisson equations with forcing (see \cite{Sato}). We will see that the appropriate phase spaces for such Lagrangian and Hamiltonian mechanics are, respectively, $\mathfrak{g}\times G\times \mathfrak{g}$ and $\mathfrak{g}^*\times G\times \mathfrak{g}^*$. 

First, consider a Lagrangian $\boldsymbol{L}: TG \times TG \rightarrow \mathbb{R}$ and the left-action $\Phi: G \times G \to G$, $\Phi_{g'}(g) = g' g = {\mathcal L}_{g'} g$, and its tangent lift 
$\widehat{\Phi} : G \times TG \rightarrow TG$ given by 
\begin{equation*}
	\widehat{\Phi}_{g'}(v_g)= T_g \Phi_{g'} (v_g) = g' v_g \in T_{g' g} G
\end{equation*}
and the corresponding diagonal action $\widehat{\Phi}^{\times}_{\tilde{g}}: TG \times TG \rightarrow TG \times TG$ defined by: 
\begin{equation*}
	\widehat{\Phi}^{\times}_{g'} (v_{g}, \tilde{v}_{\tilde{g}})= (g' v_{g}, g' \tilde{v}_{\tilde{g}})
\end{equation*}
Assuming that $\boldsymbol{L}$ is $\widehat{\Phi}^{\times}$-invariant we deduce that
\begin{equation*}
	\boldsymbol{L}(v_{g}, \tilde{v}_{\tilde{g}}) = \boldsymbol{L}( g^{-1} v_{g}, g^{-1} \tilde{v}_{\tilde{g}}),
\end{equation*}
which lets us define the reduced Lagrangian $\boldsymbol{l}: \mathfrak{g} \times G \times \mathfrak{g} \rightarrow \mathbb{R}$ by
\begin{equation*}
	\boldsymbol{l}(\eta, U, \psi) = \boldsymbol{L}(\eta, U \psi)
\end{equation*}
where $\eta = g^{-1} v_{g}, \psi = \tilde{g}^{-1} \tilde{v}_{\tilde{g}} \in \mathfrak{g}$ and $U = g^{-1} \tilde{g} \in G$.

We have the following
\begin{theorem}\label{puy}
The Euler-Lagrange equations for $\boldsymbol{L}$ are equivalent to the reduced Euler-Lagrange equations for $\boldsymbol{l}: \mathfrak{g}\times G\times \mathfrak{g}\rightarrow \mathbb{R}$: 
\begin{equation}
\begin{aligned}
	\frac{\mathrm{d}}{\mathrm{d}t}\left( \frac{\partial \boldsymbol{l}}{\partial \eta}\right) &= \mathrm{ad}^*_{\eta} \frac{\partial \boldsymbol{l}}{\partial \eta} - R_{U}^*\frac{\partial \boldsymbol{l}}{\partial U}\\
	\frac{\mathrm{d}}{\mathrm{d}t}\left( \frac{\partial \boldsymbol{l}}{\partial \psi}\right) &= \mathrm{ad}^*_{\psi} \frac{\partial \boldsymbol{l}}{\partial \psi} + L_{U}^*\frac{\partial \boldsymbol{l}}{\partial U}\\
	\frac{\mathrm{d} U}{\mathrm{d} t} &= U\psi - \eta U\\
                  &= U(\psi - \mathrm{Ad}_{U^{-1}}\eta)\,.
\end{aligned}
\label{eq:Poisson_groupoid_Euler_Lagrange}
\end{equation}
\end{theorem}

\begin{proof}
Define the functional:
\begin{equation*}
	{\mathcal J}[(\eta, U, \psi)] = \int^T_0 \boldsymbol{l}(\eta(t), U(t), \psi(t))\; dt
\end{equation*}
for some $T \in \mathbb{R} > 0$. Its critical points are the solutions of the corresponding Euler-Lagrange equations. Taking variations
\begin{equation*}
	\mathrm{d} {\mathcal J}[(\eta, U, \psi)]((\delta \eta, \delta U, \delta \psi)) = \int^T_0 \left[ \frac{\partial \boldsymbol{l}}{\partial \eta}\delta \eta+\frac{\partial \boldsymbol{l}}{\partial U}\delta U + \frac{\partial \boldsymbol{l}}{\partial \psi}\delta \psi\right]\; dt=0\,.
\end{equation*}

We know that  $\eta = g^{-1} \dot{g}$ and $\psi = \tilde{g}^{-1} \dot{\tilde{g}}$ and $U = g^{-1} \tilde{g}$. 
Therefore: 
\begin{align*}
	\delta \eta &= g^{-1}\delta\dot{g} - \Sigma \eta\\
	\delta U &= U \widetilde{\Sigma} - \Sigma U\\
	\delta \psi &= \tilde{g}^{-1}\delta\dot{\tilde{g}} - \widetilde{\Sigma}\psi
\end{align*}
where $\Sigma = g^{-1}\delta {g}$ and $\widetilde{\Sigma} = \tilde{g}^{-1}\delta \tilde{g}$. Also,
\begin{align*}
\dot{\Sigma} &= g^{-1}\delta\dot{g} - \eta \Sigma\\
\dot{\widetilde{\Sigma}} &= \tilde{g}^{-1}\delta\dot{\tilde{g}} - \psi \widetilde{\Sigma}
\end{align*}
and, in consequence,
\begin{align*}
\delta \eta &= [\eta, \Sigma] + \dot{\Sigma}\\
\delta U &= U \widetilde{\Sigma} -\Sigma U\\
\delta \psi &= [\psi, \widetilde{\Sigma}] + \dot{\widetilde{\Sigma}}\,.
\end{align*}

Since $\Sigma$ and $\widetilde{\Sigma}$ are arbitrary, using integration by parts, we deduce the equations. 
\end{proof}

A trivial example of such Lagrangians is given as follows. Consider a $\widehat{\Phi}$-invariant Lagrangian $L: TG \to \mathbb{R}$ with reduced Lagrangian $l: \mathfrak{g} \to \mathbb{R}$ and define a new Lagrangian $\boldsymbol{L}: TG \times TG \to \mathbb{R}$ as
\begin{equation*}
\boldsymbol{L}(v_{g},\tilde{v}_{\tilde{g}}) = L(\tilde{v}_{\tilde{g}}) - L(v_{g}).
\end{equation*}
Then, by (left-)trivialization we find that its associated reduced Lagrangian $\boldsymbol{l}: \mathfrak{g} \times G \times \mathfrak{g} \to \mathbb{R}$ takes the form
\begin{equation*}
\boldsymbol{l}(\eta, U, \psi) = l(\psi) - l(\eta).
\end{equation*}

According to the results of theorem \ref{puy} the equations of motion for this class of Lagrangians are simply two uncoupled and independent Euler-Poincar\'e equations.

A more general class of Lagrangians are those of the form
\begin{equation*}
\boldsymbol{l}_{\boldsymbol{k}}(\eta, U, \psi) = l(\psi) - l(\eta) - \boldsymbol{k}(\eta, U, \psi),
\end{equation*}
where $\boldsymbol{k}: \mathfrak{g} \times G \times \mathfrak{g} \to \mathbb{R}$ acts as a \emph{generalized potential}. As it will become clear in the next section, if we still want to recuperate unique and clear dynamics on $\mathfrak{g}$, it will be crucial that these Lagrangians and potentials satisfy the discrete symmetry
\begin{equation*}
\boldsymbol{k}(\psi, U^{-1}, \eta) = -\boldsymbol{k}(\eta, U, \psi),
\end{equation*}
which will result in two copies of the same dynamics when we restrict to initial conditions on the subset $(\eta, e, \eta)$, that is, the restricted vector field they define projects onto $\mathfrak{g}$.

Our aim now is to obtain a generalized potential whose contribution to the dynamics on the aforementioned subset coincides with that of a given forcing term, $f$. To do so first consider the exponential map $\exp: \mathfrak{g}\rightarrow G$. We choose $\exp$ but it is possible to take any other retraction map (see, for instance,  \cite{MR2496560} for different retraction maps). It is well known that for matrix Lie groups
\begin{equation*}
	\exp \xi = \sum_{n = 0}^\infty \frac{1}{n!} \xi^n
\end{equation*}
Obviously $\exp 0 = I$ and $T_0 \exp = \mathrm{Id}$ with the usual identifications. If we restrict ourselves to a neighborhood of the identity of the group, $\mathcal{U}_e$, then its inverse is well-defined.

With this we may then construct the function
$\boldsymbol{k}_f: \mathfrak{g} \times G \times \mathfrak{g} \rightarrow \mathbb{R}$ by
\begin{equation*}
	\boldsymbol{k}_f (\eta, U, \psi) = \frac{1}{2} \left( \left\langle f(\psi), \exp^{-1} U^{-1}\right\rangle - \left\langle f(\eta), \exp^{-1} U\right\rangle\right)
\end{equation*}
where $U$ is assumed to be in $\mathcal{U}_e$.

\begin{proposition}
Let $(l,f)$ be a regular Lagrangian system with forcing given by $l: \mathfrak{g}\rightarrow \mathbb{R}$ and $f: \mathfrak{g} \rightarrow \mathfrak{g}^*$, and define the Lagrangian system $\boldsymbol{l}_{f}: \mathfrak{g} \times G\times \mathfrak{g}\rightarrow \mathbb{R}$ by
\begin{equation*}
\boldsymbol{l}_{f} (\eta, U, \psi) = l(\psi) - l(\eta) - \boldsymbol{k}_f(\eta, U, \psi).
\end{equation*}
Then we have that the following are equivalent: 
\begin{itemize}
	\item $\sigma: I \subseteq \mathbb{R} \rightarrow \mathfrak{g}$ is a solution of the  Euler-Poincar\'e equations with forcing
	\begin{equation*}
		\frac{\mathrm{d}}{\mathrm{d}t} \left(\frac{ \partial l}{\partial \eta}\right) = \mathrm{ad}_{\eta}^*\frac{\partial l}{\partial \eta} + f(\eta)\,.
	\end{equation*}
	\item $\tilde{\sigma}:  I \subseteq \mathbb{R} \rightarrow \mathfrak{g} \times G\times \mathfrak{g}$, that is, $\tilde{\sigma}(t)=(\sigma(t), e, \sigma(t))$ is a solution of the Euler-Lagrange equations for $\boldsymbol{l}_{f}$: 
		\begin{align*}
			\frac{\mathrm{d}}{\mathrm{d}t}\left( \frac{\partial \boldsymbol{l}_{f}}{\partial \eta}\right)-\mathrm{ad}^*_{\eta}  \frac{\partial \boldsymbol{l}_{f}}{\partial \eta}+ R_{U}^*\frac{\partial \boldsymbol{l}_{f}}{\partial U} &= 0\\
			\frac{\mathrm{d}}{\mathrm{d}t}\left( \frac{\partial \boldsymbol{l}_{f}}{\partial \psi}\right)-\mathrm{ad}^*_{\psi}  \frac{\partial \boldsymbol{l}_{f}}{\partial \psi} -L_{U}^*\frac{\partial \boldsymbol{l}_{f}}{\partial U} &= 0\,.
		\end{align*}
 \end{itemize}
\end{proposition}

\begin{proof}
Applying theorem \ref{puy} to $\boldsymbol{l}_{f}$ we get
		\begin{align*}
			-\frac{\mathrm{d}}{\mathrm{d}t}\left( \frac{\partial l}{\partial \eta} - \frac{\partial \boldsymbol{k}_f}{\partial \eta}\right) + \mathrm{ad}^*_{\eta}  \left(\frac{\partial l}{\partial \eta} - \frac{\partial \boldsymbol{k}_f}{\partial \eta}\right) - R_{U}^*\frac{\partial \boldsymbol{k}_f}{\partial U} &= 0\\
			\frac{\mathrm{d}}{\mathrm{d}t}\left( \frac{\partial l}{\partial \psi} - \frac{\partial \boldsymbol{k}_f}{\partial \psi}\right) - \mathrm{ad}^*_{\psi}  \left(\frac{\partial l}{\partial \psi} - \frac{\partial \boldsymbol{k}_f}{\partial \psi}\right) + L_{U}^*\frac{\partial \boldsymbol{k}_f}{\partial U} &= 0\,.\\
		\end{align*}
Taking into account that $\exp^{-1}(e) = 0$ and $T_e \hbox{exp}^{-1} = \mathrm{Id}$, it is not difficult to see that on $(\eta, e, \eta)$ the only surviving term from $\boldsymbol{k}_f$ is
\begin{equation*}
\frac{\partial \boldsymbol{k}_f}{\partial U}(\eta, e, \eta) = - f(\eta).
\end{equation*}
Thus, on $(\eta, e, \eta)$ the Euler-Lagrange equations for $\boldsymbol{l}_{f}$ reduce to two copies of the Euler-Poincar\'e equations with forcing, which proves our claim.
\end{proof}

\section{Geometric interlude: Poisson groupoids}

\subsection{Lie groupoids and algebroids}\label{algebroide-grupoide}
First of all, we will recall some definitions related to Lie groupoid and Lie algebroids
(for more details, see \cite{mackenzie87,MR1747916}).

\begin{definition}
	A groupoid over a set $Q$ is a set $G$ together with the
	following structural maps:
	\begin{itemize}
		\item A pair of maps $\alpha: G \to Q$, the {\sl source}, and
		$\beta: G \to Q$, the {\sl target}. Thus, we can think an element $g \in G$  an arrow from $x= \alpha(g)$ to $y = \beta(g)$ in $Q$
		\begin{center}
			\begin{tikzcd}
				\underset{x = \alpha(g)}{\bullet} \arrow[r, bend left=45, "g"]
				& \underset{y = \beta(g)}{\bullet}
			\end{tikzcd}
		\end{center}
		The source and target mappings define the set of composable pairs
		\begin{equation*}
			G_{2} = \left\lbrace (g_1,g_2) \in G \times G \,\vert\, \beta(g_1) = \alpha(g_2)\right\rbrace.
		\end{equation*}
		\item A {\sl multiplication} on composable elements $\nu: G_{2} \to G$, denoted simply by $\nu(g_1,g_2) = g_1 g_2$, such that
		\begin{itemize}
			\item $\alpha(g_1 g_2) = \alpha(g_1)$ and $\beta(g_1 g_2) = \beta(g_2)$.
			\item $g_1(g_2 g_3) = (g_1 g_2) g_3$.
		\end{itemize}
		\begin{center}
			\begin{tikzcd}
				\underset{x = \alpha(g_1)}{\bullet} \arrow[r, bend left=30, "g_1"]
													\arrow[rr, bend right=30, "g_1 g_2"']
				& \underset{y = \beta(g_1) = \alpha(g_2)}{\bullet} \arrow[r, bend left=30, "g_2"]
				& \underset{z = \beta(g_2)}{\bullet}
			\end{tikzcd}
		\end{center}
		\item An {\sl inversion map} $\iota: G \to G$, to be denoted simply by $\iota(g) = g^{-1}$, such that
		\begin{itemize}
			\item $g^{-1} g = \epsilon(\beta(g))$ and $g g^{-1} = \epsilon(\alpha(g))$.
		\end{itemize}
		\begin{center}
			\begin{tikzcd}
				\underset{x = \alpha(g) = \beta(g^{-1})}{\bullet} \arrow[r, bend left=30, "g"]
				& \underset{y = \beta(g) = \alpha(g^{-1})}{\bullet} \arrow[l, bend left=30, "g^{-1}",
														         		start anchor={[xshift=4.5ex]},
														         		end anchor={[xshift=-4.5ex]}]
			\end{tikzcd}
		\end{center}
		\item An identity section $\epsilon: Q \to G$ of $\alpha$ and $\beta$, such that
		\begin{itemize}
			\item $\epsilon(\alpha(g))g = g$ and $g\epsilon(\beta(g)) = g$.
		\end{itemize}
		\begin{center}
			\begin{tikzcd}
				\underset{x = \alpha(g) = \beta(g)}{\bullet} \arrow[loop above, out=120, in=60,
																looseness=10, "g = \epsilon(x)"]
			\end{tikzcd}
		\end{center}
	\end{itemize}
\end{definition}
A groupoid $G$ over a set $Q$ will be denoted simply by the symbol
$G \rightrightarrows Q$.

The groupoid $G \rightrightarrows Q$ is said to be a {\sl Lie groupoid} if $G$ and $Q$ are differentiable manifolds and all the structural maps are differentiable with $\alpha$ and $\beta$ differentiable
submersions. If $G \rightrightarrows Q$ is a Lie groupoid then $\nu$ is a submersion, $\epsilon$ is an immersion and $\iota$ is a diffeomorphism. Moreover, if $x \in Q$, $\alpha^{-1}(x)$ (resp., $\beta^{-1}(x)$) will be said the $\alpha$-fiber (resp., the $\beta$-fiber) of $x$.

Typical examples of Lie groupoids are: the pair or banal groupoid $Q \times Q$ over $Q$ (the example that we have used along our former paper), a Lie group $G$ (as a Lie groupoid over a single point), the Atiyah groupoid $(Q \times Q)/G$ (over $Q/G$) associated with a free and proper action of a Lie group $G$ on $Q$... (see \cite{mackenzie87}).

\begin{definition}
	If $G \rightrightarrows Q$ is a Lie groupoid and $g \in G$ then the left-translation by
	$g \in G$ and the right-translation by $g$ are the
	diffeomorphisms
	\begin{equation*}
		\begin{array}{lll}
		\mathcal{L}_{g}: \alpha^{-1}(\beta(g)) \longrightarrow \alpha^{-1}(\alpha(g)), \quad & g'\longrightarrow \mathcal{L}_{g}(g') = gg',\\
		\mathcal{R}_{g}: \beta^{-1}(\alpha(g)) \longrightarrow \beta^{-1}(\beta(g)), \quad & g' \longrightarrow \mathcal{R}_{g}(g') = g'g.
		\end{array}
	\end{equation*}
\end{definition}
Note that $\mathcal{L}_{g}^{-1} = \mathcal{L}_{g^{-1}}$ and $\mathcal{R}_{g}^{-1} = \mathcal{R}_{g^{-1}}$.

\begin{definition}
	A vector field $\xi\in {\mathfrak X}(G)$ is said to be
	{\sl left-invariant} (resp., right-invariant) if it is
	tangent to the fibers of $\alpha$ (resp., $\beta$) and
	${\xi}(gg') = T_{g'}\mathcal{L}_{g} {\xi}(g')$ (resp.,
	${\xi}(gg') = T_{g}\mathcal{R}_{g'} {\xi}(g))$, for $(g,g') \in
	G_{2}$.
\end{definition}

The infinitesimal version of a Lie groupoid is a Lie algebroid which is defined as follows.

\begin{definition}
	A {\sl Lie algebroid} is a real vector bundle $A\rightarrow Q$ equipped  with a Lie bracket 
	$\lcf \cdot ,\cdot  \rcf$ on its sections $\Gamma(A)$ and a bundle map $\rho: A\rightarrow TQ$ called the {\sl anchor map} such that 
	the homomorphism of $C^{\infty}(Q)$-modules induced by the anchor map, that we also denote by $\rho: \Gamma(A)\rightarrow {\mathfrak X}(Q)$, verifies  
	\[
	\lcf X, f Y\rcf =f\lcf X, Y\rcf +\rho(X)(f) Y,
	\]
	for $X, Y\in \Gamma (A)$ and $f\in C^{\infty}(Q)$.
\end{definition}

With this definition the anchor map $\rho: \Gamma(A)\rightarrow {\mathfrak X}(Q)$
is a Lie algebra homomorphism, where ${\mathfrak X}(Q)$ is endowed with the usual Lie bracket of vector field $\left[\cdot, \cdot\right]$.  

\begin{definition}
	Given a Lie groupoid $G \rightrightarrows Q$, the {\sl associated Lie algebroid} $AG\rightarrow Q$ is given by its fibers $A_{q}G = V_{\epsilon(q)}\alpha = \ker (T_{\epsilon(q)}\alpha)$. There is a bijection between the space $\Gamma (AG)$ and the set of	left-invariant vector fields on $G$. If $X$ is a section of $\tau: AG \to Q$, the corresponding left-invariant  vector field on $G$ will be denoted $\lvec{X}$ (resp., $\rvec{X}$), where
	\begin{equation}\label{linv}
		\lvec{X}(g) = (T_{\epsilon(\beta(g))}\mathcal{L}_{g}) \, X(\beta(g)),
	\end{equation}
for $g \in G$. Using the above facts, one may introduce a bracket $\lcf\cdot , \cdot\rcf$ on the space of sections $\Gamma(AG)$ and a bundle map $\rho: AG \to TQ$, which are defined by
	\begin{equation}\label{LA}
		\lvec{\lcf X, Y\rcf} = [\lvec{X}, \lvec{Y}], \makebox[.3cm]{}
		\rho(X)(q) = (T_{\epsilon(q)}\beta) \, X(q),
	\end{equation}
for $X, Y \in \Gamma(AG)$ and $q \in Q$. 
\end{definition}

Using that $[\cdot, \cdot]$ induces a Lie algebra structure on the space of vector fields on $G$, it is easy to prove that $\lcf \cdot, \cdot \rcf$ also defines a Lie algebra structure on $\Gamma(AG)$. In addition, it follows that
\begin{equation*}
	\lcf X, f Y\rcf = f \lcf X, Y\rcf + \rho(X)(f) Y,
\end{equation*}
for $X, Y \in \Gamma(AG)$ and $f \in C^{\infty}(Q)$.

One can also stablish a bijection between sections $X\in\Gamma(AG)$ and right invariant vector fields $\rvec{X}\in {\mathfrak X}(G)$ defined by 
\begin{equation}\label{rinv}
	\rvec{X}(g) = - (T_{\epsilon(\alpha(g))}\mathcal{R}_{g} \, T_{\epsilon(\alpha(g))}\iota) \, X(\alpha(g)),
\end{equation}
which yields the Lie bracket relation
\begin{equation*}
	\rvec{\lcf X, Y\rcf} =- [\rvec{X}, \rvec{Y}].
\end{equation*}

The following Proposition will be useful for the results in this paper. 

\begin{proposition}\label{cft} (See \cite{Sato}).
	Let $G \rightrightarrows Q$ be a Lie groupoid and $Z\in \mathfrak{X}(G)$ a vector field invariant by the inversion, that is, 
	\begin{equation*}
		T_g\iota \, Z(g) = Z(g^{-1}), \text{ for all } g \in G
	\end{equation*}
	Then, for all $q \in Q$, 
	\begin{equation*}
		Z(\epsilon(q)) \in T_{\epsilon(q)} \epsilon (Q).
	\end{equation*}
\end{proposition}

\subsection{Poisson groupoids}
In the following, we introduce the notion of Poisson groupoid (see \cite{weipoisson}). 
\begin{definition}
	A {\sl Poisson groupoid} is a Lie groupoid $\Gamma \rightrightarrows Q$, such that
	\begin{enumerate}
		\item $(\Gamma, \left\{\; ,\; \right\})$ is a Poisson manifold;
		\item the graph of $\nu: \Gamma_2 \rightarrow \Gamma$ is a coisotropic submanifold   $\Gamma\times \Gamma\times \Gamma^-$, where $\Gamma^- = (\Gamma, -\left\{\; ,\; \right\})$ has the negative Poisson  structure.
	\end{enumerate}
\end{definition}

Some interesting properties of Poisson groupoids are: 

\begin{theorem} Let $\Gamma$ be a Poisson groupoid
	\begin{enumerate}
		\item The identity section is coisotropic in $\Gamma$. 
		\item The inversion $\iota$ is an anti-Poisson morphism. 
		\item There is a unique Poisson structure on $\Gamma_0$ for which $\alpha$ is a Poisson mapping (and $\beta$ is an anti-Poisson morphism). 
	\end{enumerate}

\end{theorem}

In our work, we are interested on a concrete  example  of Poisson groupoid (see, for instance, \cite{MR1612160}). Let $G$ be a Lie group and $\mathfrak{g}$ its Lie algebra. The manifold $\mathfrak{g}^*\times G\times \mathfrak{g}^*$ has a natural structure of Lie groupoid  where the structural functions are
\begin{equation}\label{groupoid_struct}
	\begin{alignedat}{4}
		\alpha(\lambda_1, U, \lambda_2) &= \lambda_1,  &\iota (\lambda_1, U, \lambda_2) &= (\lambda_2, U, \lambda_1),\\
		\beta(\lambda_1, U, \lambda_2) &= \lambda_2, \quad   &\epsilon (\lambda) &= (\lambda, e, \lambda),\\
	&\nu((\lambda_1, U, \lambda_2), (\lambda_2, V, \lambda_3)) &= (\lambda_1, UV, \lambda_3), 
	\end{alignedat}
\end{equation}
and the Poisson bracket is given by: 
\begin{equation}\label{poisson_struct}
	\begin{alignedat}{4}
		\left\lbrace\Xi_1, \Xi_2\right\rbrace(\lambda_1, U, \lambda_2) &= 0, \quad & \left\lbrace\Xi_1, F \right\rbrace &= -\overrightarrow{\xi} F,\\
		\left\lbrace\Xi_1, \Xi_1'\right\rbrace(\lambda_1, U, \lambda_2) &= \left\langle \lambda_1, [\xi, \xi']\right\rangle, \quad &\left\lbrace\Xi_2, F\right\rbrace &= - \overleftarrow{\xi} F.\\
		\left\lbrace\Xi_2, \Xi_2'\right\rbrace(\lambda_1, U, \lambda_2) &= -\left\langle \lambda_2, [\xi, \xi']\right\rangle, \quad & & \vphantom{\overrightarrow{\xi}}
	\end{alignedat}
\end{equation}
where $\xi, \xi' \in \mathfrak{g}$ induce the functions given by
\begin{alignat*}{4}
	\Xi_1(\lambda_1, U, \lambda_2) &= \langle \lambda_1, \xi\rangle, \quad & \Xi_1'(\lambda_1, U, \lambda_2) &= \langle \lambda_1, \xi'\rangle,\\
	\Xi_2(\lambda_1, U, \lambda_2) &= \langle \lambda_2, \xi\rangle, \quad & \Xi_2'(\lambda_1, U, \lambda_2) &= \langle \lambda_2, \xi'\rangle,
\end{alignat*}
and $F: \mathfrak{g}^*\times G\times \mathfrak{g}^* \to \mathbb{R}$ is the pull-back of a function on the Lie group $G$. The Lie groupoid $\mathfrak{g}^*\times G\times \mathfrak{g}^*$ equipped with this bracket is a Poisson groupoid, and observe that the linear Poisson bracket is completely determined by these functions. In some ocassions we will identify the function $\Xi_i \equiv \lambda_i$, $i = 1,2$ when there is no possible confusion.

$\Gamma = \mathfrak{g}^*\times G\times \mathfrak{g}^*$ equipped with this bracket is a Poisson groupoid (see \cite{MR1612160}). 

%In our paper, since we are working with equations reduced by left translation we will modify this Poisson structure using the isomorphism: 
%\[
%\begin{array}{rrcl}
%\Psi:& \mathfrak{g}^*\times G\times \mathfrak{g}^*&\longrightarrow &  \mathfrak{g}^*\times G\times \mathfrak{g}^*\\
%& (\mu_1, U, \mu_2) &\longmapsto & (\mu_1, U, -\mu_2)
%\end{array}
%\]       
%Now, the Lie groupoid structure reads as
%\begin{eqnarray*}
%	\alpha_{\pm} (\mu_1, U, \mu_2)&=&- \mu_1\; \quad \beta_{\pm}(\mu_1, U, \mu_2)= \mu_2\\
%	\epsilon_{\pm} (\mu_1)&=& (-\mu_1, e, \mu_1), \quad  \iota_{\pm} (\mu_1, U, \mu_2)= (-\mu_2, U^{-1}, -\mu_1)\\
%	m_{\pm}((\mu_1, U_1, \mu_2), (-\mu_2, U_2, \mu_3))&=& (\mu_1, U_1U_2, \mu_3)
%\end{eqnarray*}
%and Poisson structure is now modified as follows 
%\begin{eqnarray}\label{poio}
%	\{\xi, \eta\}_{\pm}(\mu_1, U, \mu_2)&=&0\quad , \{\xi, \xi'_1\}_{\pm}(\mu_1, U, \mu_2)=-\langle \mu_1, [\xi, \xi'_1]\rangle,  \\
%	\{\eta, \xi'_2\}_{\pm}(\mu_1, U, \mu_2)&=&-\langle \mu_2, [\xi, \xi']\rangle ,\  \{\xi, f\}_{\pm}=\rvec{\xi} F\; ,\ \{\eta, f\}_{\pm}=-\lvec{\xi} F \; .
%\end{eqnarray}
%Of course $\Gamma_{\pm}=(\mathfrak{g}^*\times G\times \mathfrak{g}^*)_{\pm}$ with this Lie groupoid structure and with the modified bracket $\{ \; ,\; \}_{\pm}$ is also  a Poisson groupoid. 

The following proposition will be useful. 

\begin{proposition}\label{prop-poi}
	Let $\Gamma \rightrightarrows Q$ be a Poisson groupoid with Poisson bracket $\{\; ,\; \}$ and $E: \Gamma\rightarrow \mathbb{R}$ a function such that $E \circ \iota = - E$. Then, the corresponding Hamiltonian vector field $X_E$ defined
	\begin{equation*}
		{X_E}(F) = \left\{ F, E \right\},
	\end{equation*}
	verifies that $X_E(\epsilon(q))\in T_{\epsilon(q)}\epsilon (Q)$ for all $q\in Q$.
\end{proposition}
\begin{proof}
	We use that in a Poisson groupoid the inversion is an anti-Poisson morphism
	\begin{equation*}
		\iota^*\left\{ F, \tilde{F} \right\} = - \left\{ \iota^*F, \iota^*\tilde{F} \right\},
	\end{equation*}
for $F, \tilde{F}: \Gamma\rightarrow {\mathbb R}$, 
	that in particular implies that
	\begin{equation*}
		T\iota(X_{E}) = - X_{E \circ \iota}
	\end{equation*}
	but since $E = - E \circ \iota$ then $T\iota(X_E) = X_E$.  
	Now applying Proposition \ref{cft} we deduce that $X_E (\epsilon (q))\in T_{\epsilon(q)}\epsilon(Q)$. 
\end{proof}

%\begin{proposition}\label{prop-poi2}
%	Let $G \rightrightarrows Q$ be a Poisson  groupoid with Poisson bracket $\{\; ,\; \}$   and $E_1: G\rightarrow \mathbb{R}$  and $E_2: G\rightarrow \mathbb{R}$ two  functions such that $E_1\circ i=-E_1$ and $E_2\circ i=-E_2$ and
%	\[
%	dE_1(\epsilon(q))=	dE_2(\epsilon(q))
%	\]
	% then 		$X_{E_1}(\epsilon(q))=	X_{E_2}(\epsilon(q))$. 
	% \end{proposition}
 %\begin{proof}
 	%It  is a direct consequence of the previous Proposition and the %definition of Hamiltonian vector field.  
 %\end{proof}

\subsection{Free Hamiltonian description of forced Lie-Poisson equations}
Given a Lagrangian $\boldsymbol{l}: \mathfrak{g} \times G \times \mathfrak{g} \to \mathbb{R}$, one may immediately define the usual Legendre transformation $\mathcal{F}\boldsymbol{l}(\eta, g, \psi) = \left(\partial \boldsymbol{l}/ \partial {\eta}, U, \partial \boldsymbol{l}/ \partial {\psi}\right)$ to obtain a Hamiltonian description, but in order to maintain the natural Poisson groupoid structure in $\mathfrak{g}^* \times G \times \mathfrak{g}^*$ given by eqs. \eqref{groupoid_struct} and \eqref{poisson_struct} it is convenient to define a modified Legendre transformation
\begin{equation*}
	\begin{array}{rrcl}
	\mathcal{F}\boldsymbol{l}^{\times}:& \mathfrak{g} \times G \times \mathfrak{g} & \to & \mathfrak{g}^* \times G \times \mathfrak{g}^*\\
	& (\eta, U, \psi) & \mapsto & \left(\lambda = -\frac{\partial \boldsymbol{l}}{\partial \eta}, U, \mu = \frac{\partial \boldsymbol{l}}{\partial \psi}\right).
	\end{array}
\end{equation*}
together with a modified interior product $\left\langle (\lambda, U, \mu), (\eta, U, \psi)\right\rangle_{\times} = \left\langle \mu, \psi\right\rangle - \left\langle \lambda, \eta\right\rangle$. One may quickly check that these definitions ensure that
\begin{equation*}
\left\langle \mathcal{F}\boldsymbol{l}^{\times}(\eta, U, \psi), (\eta, U, \psi)\right\rangle_{\times} = \left\langle \mathcal{F}\boldsymbol{l}(\eta, U, \psi), (\eta, U, \psi)\right\rangle.
\end{equation*}

If the modified Legendre transformation is a local diffeomorphism then we may implicitly define the associated Hamiltonian by
\begin{align*}
\left(\boldsymbol{h} \circ \mathcal{F}\boldsymbol{l}^{\times}\right)(\eta, U, \psi) &= E_{\boldsymbol{l}} (\eta, U, \psi)\\
&= \left\langle \mathcal{F}\boldsymbol{l}^{\times}(\eta, U, \psi), (\eta, U, \psi)\right\rangle_{\times} - \boldsymbol{l}(\eta, U, \psi),
\end{align*}
whose equations of motion are
\begin{align*}
\dot{\lambda} &= - \mathrm{ad}^*_{\partial \boldsymbol{h}/\partial \lambda} \lambda - \mathcal{R}^*_{U} \frac{\partial \boldsymbol{h}}{\partial U}\\
\dot{\mu} &= \mathrm{ad}^*_{\partial \boldsymbol{h}/\partial \mu} \mu - \mathcal{L}^*_{U} \frac{\partial \boldsymbol{h}}{\partial U}\\
\dot{U} &= U \left(\frac{\partial \boldsymbol{h}}{\partial \mu} + \mathrm{Ad}_{U^{-1}} \frac{\partial \boldsymbol{h}}{\partial \lambda}\right)
\end{align*}

These equations are given  in Poisson form, $\dot{F}=\{F, \boldsymbol{h}\}$. Here,    the  Poisson bracket $\left\{\;,\;\right\}$ is defined by 
\begin{align*}
\left\lbrace A, B\right\rbrace (\lambda, U, \mu) &= - \left\langle \mu, \left[ \frac{\partial A}{\partial \mu}, \frac{\partial B}{\partial \mu}\right]\right\rangle + \left\langle \lambda, \left[ \frac{\partial A}{\partial \lambda}, \frac{\partial B}{\partial \lambda}\right]\right\rangle\\
&+ \left\langle \mathcal{L}^*_{U} \frac{\partial A}{\partial U}, \frac{\partial B}{\partial \mu} + \mathrm{Ad}_{U^{-1}} \frac{\partial B}{\partial \lambda} \right\rangle\\
&- \left\langle \mathcal{L}^*_{U} \frac{\partial B}{\partial U}, \frac{\partial A}{\partial \mu} + \mathrm{Ad}_{U^{-1}} \frac{\partial A}{\partial \lambda} \right\rangle,
\end{align*}
where $A,B: \mathfrak{g}^* \times G \times \mathfrak{g}^* \to \mathbb{R}$. This bracket   is exactly the linear Poisson structure defined on \eqref{poisson_struct}. Thus $\mathfrak{g}^*\times G\times \mathfrak{g}^*$ equipped with the Lie groupoid structure \eqref{groupoid_struct} together with this Poisson structure (and associated Poisson bi-vector $\Pi$) becomes a {\bf Poisson groupoid}. 

Let us consider first a particular case of Hamiltonians on this groupoid.
\begin{lemma}\label{aqw}
Let $h: \mathfrak{g}^* \rightarrow \mathbb{R}$ be a Hamiltonian function. Consider the Hamiltonian 
$\boldsymbol{h}: \mathfrak{g}^* \times G \times \mathfrak{g}^* \rightarrow \mathbb{R}$ defined by
\begin{equation*}
	\boldsymbol{h}(\lambda, U, \mu) = h(\mu) - h(\lambda)
\end{equation*}
then 
\begin{enumerate}
	\item $\sharp^{\Pi}({d\boldsymbol{h}})= X_{\boldsymbol{h}}$ is tangent to $\epsilon(\mathfrak{g}^*)$;
	\item $\left.{X_{\boldsymbol{h}}}\right\vert_{\epsilon(\mathfrak{g}^*)} = \epsilon_*(X_{h})$.
\end{enumerate}
\end{lemma}
\begin{proof}
	For the proof of the first part, observe that 
	\begin{equation*}
		(\boldsymbol{h} \circ \iota)(\lambda, U, \mu) = \boldsymbol{h}(\mu, U^{-1}, \lambda) = h(\lambda) - h(\mu) = - \boldsymbol{h}(\lambda, U, \mu)	
	\end{equation*}
	and apply Proposition \ref{prop-poi}.
	
	For the second part, it is easy to check using expressions (\ref{poisson_struct}) that 
	\begin{align*}
		\left\{\eta, \boldsymbol{h}\right\}(\lambda, U, \mu) &= \left\langle \lambda, \left[ \eta, \frac{\partial \boldsymbol{h}}{\partial \lambda}\right]\right\rangle = - \left\langle \lambda, \left[ \eta, h'(\lambda)\right]\right\rangle\\
		\left\{\psi, \boldsymbol{h}\right\}(\lambda, U, \mu) &= - \left\langle \mu, \left[ \eta, \frac{\partial \boldsymbol{h}}{\partial \mu}\right]\right\rangle = - \left\langle \mu, \left[ \eta, h'(\mu)\right]\right\rangle\\
		\left\{F, \boldsymbol{h}\right\}(\lambda, U, \mu) &= \left\langle {\mathcal L}^*_{U} \frac{\partial F}{\partial U}, \frac{\partial \boldsymbol{h}}{\partial \mu} + \mathrm{Ad}_{U^{-1}} \frac{\partial \boldsymbol{h}}{\partial \lambda}\right\rangle\\
		&= \left\langle {\mathcal L}^*_{U} \frac{\partial F}{\partial U}, h'(\mu) - \mathrm{Ad}_{U^{-1}} h'(\lambda)\right\rangle
	\end{align*}
 where $F: \mathfrak{g}^*\times G\times \mathfrak{g}^* \to \mathbb{R}$ is the pull-back of a function on the Lie group $G$.

Therefore,  if $F: \mathfrak{g}^*\times G\times \mathfrak{g}^* \to \mathbb{R}$ is the pull-back of a function on the Lie group $G$ we have that 
\begin{align*}
	(X_{\boldsymbol{h}})_{(\mu, e, \mu)}(\eta) &= -\left\langle \mu, \left[ \eta, h'(\mu)\right] \right\rangle,\\
	(X_{\boldsymbol{h}})_{(\mu, e, \mu)}(\psi) &= -\left\langle \mu, \left[ \eta, h'(\mu)\right] \right\rangle,\\
	(X_{\boldsymbol{h}})_{(\mu, e, \mu)}(F) &= 0,
\end{align*}
	which is exactly the same as $\epsilon_*(X_h)$ since
\begin{equation*}
\epsilon_*(X_h)(\mu, e, \mu) =	
\left(\mathrm{ad}^*_{h'(\mu)} \mu, 0, \mathrm{ad}^*_{h'(\mu)}\mu\right) \in \mathfrak{g}^*\times \mathfrak{g}\times \mathfrak{g}^* \equiv 
T_{(\mu, e, \mu)} (\mathfrak{g}^* \times G \times \mathfrak{g}^*).\qedhere
\end{equation*}

\end{proof}

Our aim is to generalize lemma \ref{aqw} for the case of Lie-Poisson systems with forcing, that is, we have Hamiltonian function $h: \mathfrak{g}^* \rightarrow \mathbb{R}$ and the force expressed by $\tilde{f}: \mathfrak{g}^*\rightarrow  \mathfrak{g}^*$, both determining the Lie-Poisson equations with forcing
\begin{equation*}
	\dot{\mu} = \mathrm{ad}_{h'(\mu)}^*\mu + \tilde{f}(\mu),
\end{equation*}
which define the vector field 
\begin{equation}
	Y_{h,\tilde{f}}(\mu) = X_h(\mu) + \sharp^{\Pi}(\tilde{f})(\mu) \in T_{\mu}\mathfrak{g}^* \equiv \mathfrak{g}^*.
\end{equation}

Similar to the Lagrangian case let us define a function $\boldsymbol{k}_{\tilde{f}}: \mathfrak{g}^* \times G \times \mathfrak{g}^* \rightarrow \mathbb{R}$ by
\begin{equation*}
	\boldsymbol{k}_{\tilde{f}}(\lambda, U, \mu) = \frac{1}{2} \left( \left\langle \tilde{f}(\mu), \exp^{-1} U^{-1}\right\rangle - \left\langle \tilde{f}(\lambda), \exp^{-1} U\right\rangle\right)
\end{equation*}
where $U$ is assumed to be in a neighborhood $\mathcal{U}_e$ of the identity element $e\in G$. With this we can state the following theorem.

\begin{theorem}\label{main}
Let $h: \mathfrak{g}^*\rightarrow \mathbb{R}$ be a Hamiltonian function and $\tilde{f}: \mathfrak{g}^*\rightarrow \mathfrak{g}$ representing an external force . Consider the Hamiltonian 
$\boldsymbol{h}_{\tilde{f}}: \mathfrak{g}^* \times G\times \mathfrak{g}^* \rightarrow \mathbb{R}$ defined by
\begin{equation*}
\boldsymbol{h}_{\tilde{f}}(\lambda, U, \mu) = h(\mu) - h(\lambda) + \boldsymbol{k}_{\tilde{f}}(\lambda, U, \mu)
\end{equation*}
then 
\begin{enumerate}
	\item $\sharp^{\Pi}({d\boldsymbol{h}_{\tilde{f}}})= X_{\boldsymbol{h}_{\tilde{f}}}$ is tangent to $\epsilon(\mathfrak{g}^*)$;
	\item $\left.{X_{\boldsymbol{h}_{\tilde{f}}}}\right\vert_{\epsilon(\mathfrak{g}^*)} = \epsilon_*\left(Y_{h,\tilde{f}}\right)$.
\end{enumerate}

\end{theorem}
\begin{proof}

The proof follows the same steps as those of lemma \ref{aqw}. For the first part observe that 
\begin{align*}
(\boldsymbol{k}_{\tilde{f}} \circ \iota)(\lambda, U, \mu) &= \boldsymbol{k}_{\tilde{f}}(\mu, U^{-1}, \lambda)\\
&= \frac{1}{2} \left( \left\langle \tilde{f}(\lambda), \exp^{-1} U \right\rangle - \left\langle \tilde{f}(\mu), \exp^{-1} U^{-1}\right\rangle \right)\\
&= -\boldsymbol{k}_{\tilde{f}}(\lambda, U, \mu).
\end{align*}
For the second part, if one takes into account that $\exp^{-1}(e) = 0$ and $T_e \hbox{exp}^{-1} = \mathrm{Id}$ then it is not difficult to see that
\begin{align*}
	\left({X_{\boldsymbol{h}_{\tilde{f}}}}\right)_{(\mu, e, \mu)}(\eta) &= -\left\langle \mu, \left[ \eta, h'(\mu)\right] \right\rangle + \left\langle \tilde{f}(\mu), \eta \right\rangle,\\
	\left({X_{\boldsymbol{h}_{\tilde{f}}}}\right)_{(\mu, e, \mu)}(\psi) &= -\left\langle \mu, \left[ \eta, h'(\mu)\right] \right\rangle + \left\langle \tilde{f}(\mu), \eta \right\rangle,\\
	\left({X_{\boldsymbol{h}_{\tilde{f}}}}\right)_{(\mu, e, \mu)}(F) &= 0.
\end{align*}
which coincides with $\epsilon_*(Y_{h,\tilde{f}})$
\end{proof}

\begin{proposition}\label{inverseLegendre}
Let $(h, \tilde{f})$ be a regular Hamiltonian system with forcing. Then its associated Hamiltonian $\boldsymbol{h}_{\tilde{f}}$ is regular in a neighborhood of $\epsilon(\mathfrak{g}^*)$
\end{proposition}
\begin{proof}
	
Observe that the transformation
\begin{equation*}
\begin{array}{rrcl}
\mathcal{F} \boldsymbol{h}_{\tilde{f}}^{\times}: & \mathfrak{g}^* \times G \times \mathfrak{g}^* &\rightarrow & \mathfrak{g} \times G \times \mathfrak{g}\\
&\left(\lambda, U, \mu \right) & \mapsto & \left(-\frac{\partial \boldsymbol{h}_{\tilde{f}}}{\partial \lambda}, U, \frac{\partial \boldsymbol{h}_{\tilde{f}}}{\partial \mu}\right)
\end{array}
\end{equation*}
reduces to $\mathcal{F} \boldsymbol{h}_{\tilde{f}}^{\times}(\mu, e, \mu) = \left(h'(\mu), e, h'(\mu)\right)$ at the identity set. Then it must be a local diffeomorphism in a neighborhood of this set since its Hessian matrix
\begin{equation*}
\left(\begin{array}{cc}
h''(\mu) & 0\\
0 & h''(\mu)
\end{array}\right)
\end{equation*}
is regular on $\epsilon(\mathfrak{g}^*)$ and therefore regular on a neighbourhood of it.
\end{proof}

%We are now ready to state the main theorem of this section and it will be important to the construction of the exact discrete Lagrangian for forced systems.

\begin{theorem}
Let $(l, f)$ and $(h, \tilde{f})$ be a regular forced Lagrangian system and its associated forced Hamiltonian system, and denote by $\tilde{\boldsymbol{h}}_{\tilde{f}} = E_{\boldsymbol{l}_{f}} \circ \left(\mathcal{F}\boldsymbol{l}_{f}^{\times}\right)^{-1}$ and $\boldsymbol{h}_{\tilde{f}}$ the corresponding generalized Hamiltonians. Then their respective Hamiltonian vector fields $X_{\tilde{\boldsymbol{h}}_{\tilde{f}}}$ and $X_{\boldsymbol{h}_{\tilde{f}}}$ satisfy that $\left.X_{\tilde{\boldsymbol{h}}_{\tilde{f}}}\right\vert_{\epsilon(\mathfrak{g}^{*})} = \left.X_{\boldsymbol{h}_{\tilde{f}}}\right\vert_{\epsilon(\mathfrak{g}^{*})}$.
\end{theorem}

\begin{proof}
As in theorem \ref{main} we construct the extended Hamiltonian $\boldsymbol{h}_{\tilde{f}}$, and we note that it coincides with the one implicitly defined as
\begin{align*}
\left(\boldsymbol{h}_{\tilde{f}} \circ \mathcal{F} \boldsymbol{l}^{\times}\right) (\eta, U, \psi) &= E_{\boldsymbol{l}}(\eta, U, \psi) + \left(\tilde{\boldsymbol{k}} \circ \mathcal{F} \boldsymbol{l}^{\times}\right)(\eta, U, \psi)\\
&= \left\langle l'(\psi), \psi \right\rangle -l(\psi)- \left\langle l'(\eta), \eta \right\rangle+l(\eta)\\
&+ \frac{1}{2} \left( \left\langle f(\psi), \exp^{-1} U^{-1} \right\rangle - \left\langle f(-\eta), \exp^{-1} U\right\rangle \right),
\end{align*}
where $\boldsymbol{l}(\eta, U, \psi) = l(\psi) - l(\eta)$.

Applying the results of theorem \ref{main} we have that
\begin{equation*}
\left. \left(\mathcal{F} \boldsymbol{l}^{\times}\right)^* \mathrm{d} \boldsymbol{h}_{\tilde{f}} \right\vert_{\epsilon(\mathfrak{g})}  = (-\eta, -f(\eta),\eta) \in \mathfrak{g}\times \mathfrak{g}^* \times \mathfrak{g},
\end{equation*}

%From the regular extended Lagrangian $\boldsymbol{l}_{\boldsymbol{k}}$ we can construct the other Hamiltonian $\tilde{\boldsymbol{h}}_{\tilde{f}}$ implicitly defined as:
%\begin{equation*}
%\left(\boldsymbol{h}'_{\tilde{\boldsymbol{k}}'} \circ \mathcal{F} \boldsymbol{l}_{\boldsymbol{k}}^{\times}\right) (\eta, U, \psi) = \left\langle \mathcal{F} \boldsymbol{l}_{\boldsymbol{k}}^{\times}(\eta, U, \psi), (\eta, U, \psi)\right\rangle_{\times} -  \boldsymbol{l}_{\boldsymbol{k}}^{\times}(\eta, U, \psi)
%\end{equation*}
%This Hamiltonian is of the form
%\begin{align*}
%\tilde{\boldsymbol{h}}_{\tilde{f}}(\lambda, U, \mu) &= \left\langle \lambda, \eta(\lambda,U)\right\rangle + l(\eta(\lambda,U)) + \frac{1}{2}\left\langle f(\eta(\lambda,U)), \exp^{-1} U\right\rangle\\
%& +\langle \mu, \psi(\mu,U)\rangle - l(\psi(\mu,U)) - \frac{1}{2} \left\langle {f}(\psi(\mu,U), \exp^{-1}U^{-1}\right\rangle
%\end{align*}
%where $\eta(\lambda,U)$ and $\psi(\mu, U)$ are implicitly defined by
%\begin{align*}
%\lambda &= \frac{\partial {\mathbf L}_K}{\partial \eta} = -l'(\eta)-\frac{1}{2}\left\langle \frac{\partial f}{\partial \eta}(\eta), \exp^{-1} U \right\rangle\\
%\mu &= \frac{\partial {\mathbf L}_K}{\partial \psi} = l'(\psi)+\frac{1}{2}\left\langle \frac{\partial f}{\partial \eta}(\psi), \exp^{-1} U^{-1}\right\rangle 
%\end{align*}

From the definition of the second Hamiltonian, $\tilde{\boldsymbol{h}}_{\tilde{f}}$ and taking into account the results of proposition \ref{inverseLegendre}, it follows that
\begin{equation*}
\left. \left(\mathcal{F} \boldsymbol{l}_{f}^{\times}\right)^* \mathrm{d} \tilde{\boldsymbol{h}}_{\tilde{f}} \right\vert_{\epsilon(\mathfrak{g})}  = (-\eta, -f(\eta),\eta),
\end{equation*}
and thus $\left.\mathrm{d} \boldsymbol{h}_{\tilde{f}} \right\vert_{\epsilon(\mathfrak{g^*})} = \left.\mathrm{d} \tilde{\boldsymbol{h}}_{\tilde{f}} \right\vert_{\epsilon(\mathfrak{g^*})} = \left(-h'(\mu), -\tilde{f}(\mu), h'(\mu)\right)$, which together with the application of $\sharp^{\Pi}$ finishes our proof.
	\end{proof}

%Let $f: \mathfrak{g}^*\rightarrow {\mathbb %R}$ be a Casimir function  for the Lie %Poisson-bracket, that is  $ad^*_{\partial  %f/\partial \mu} \mu=0$, then the function
%$F_1: \mathfrak{g}^*\times G\times %\mathfrak{g}^*\rightarrow \mathbb{R}$
 %and $F_2:  \mathfrak{g}^*\times G\times %\mathfrak{g}^*\rightarrow \mathbb{R}$ defined respectively by 
%$F_1(\lambda, h, \mu)=f(\lambda)$ and %$F_2(\lambda, h, \mu)=f(\mu)$, then %$T\hbox{pr}_1 X_{F_1}=0$ and $T\hbox{pr}_2 %X_{F_2}=0$. 

%Suppose that $\varphi: \mathfrak{g}^*\times %G\times \mathfrak{g}^*\rightarrow %\mathfrak{g}^*\times G\times {\mathfrak %g}^*$ is a Poisson isomorphism for $\{\; %,\;\}_{**}$ then 
%\[
%\{\; F, G\; \}_{**}\circ \varphi=\{\; F\circ \varphi %, G\circ \varphi\; \}_{**}
%\]
%Therefore
%\[
%X_{F}=T\varphi( X_{F\circ \varphi^{-1}})
%\]

\section{Discrete case}
Consider a discrete Lagrangian $\boldsymbol{L}_d: G^4 = G \times G \times G \times G \rightarrow \mathbb{R}$ invariant under the action
\begin{equation*}
	\begin{array}{rrcl}
		\widehat{\Phi}_d :& G \times  G^4                        & \longrightarrow & G^4\\
		                & \left(g', (g_1, g_2, \tilde{g}_1, \tilde{g}_2)\right) & \longmapsto     & (g' g_1, g' g_2, g' \tilde{g}_1, g' \tilde{g}_2)
	\end{array}
\end{equation*}
and define the reduced Lagrangian $\boldsymbol{\ell}_d: G \times G \times G \rightarrow \mathbb{R}$ by 
\begin{equation*}
\boldsymbol{\ell}_d( V, U, W) = \boldsymbol{L}_d( e, V, U, U W)
\end{equation*}
where $V = g_1^{-1} g_2$, $U = g_1^{-1} \tilde{g}_1$ and $W = \tilde{g}_1^{-1} \tilde{g}_{2}$.

Then, if $V_{k} = g_k^{-1} g_{k+1}$, $U_{k} = g_k^{-1} \tilde{g}_k$ and $W_{k} = \tilde{g}_k^{-1} \tilde{g}_{k+1}$, we have that 
\begin{align*}
	\delta V_k &= (L_{V_k})_* \mathrm{P}_{k+1} - (R_{V_k})_* \mathrm{P}_k,\\
	\delta W_k &= (L_{W_k})_* \Sigma_{k+1} - (R_{W_k})_* \Sigma_k,\\
	\delta U_k &= U_k \Sigma_k - \mathrm{P}_k U_k,
\end{align*}
where $\mathrm{P}_k = g_k^{-1} \delta g_k$ and $\Sigma_k = \tilde{g}_k^{-1} \delta \tilde{g}_k$.
  
  \begin{proposition}
  	Given a discrete Lagrangian $\boldsymbol{l}_d : G\times G\times G\rightarrow \mathbb{R}$, the following are equivalent: 
  	\begin{enumerate}
 \item The discrete variational principle
\begin{equation}\label{eq:discret_red_action}
	\delta \mathcal{J}_d = \sum_{k = 0}^{N-1} \delta \boldsymbol{l}_d( V_k, U_k, W_k ) = 0
\end{equation}
holds using variations of the form $\delta V_k = (\mathcal{L}_{V_k})_* \mathrm{P}_{k+1} - (\mathcal{R}_{V_k})_* \mathrm{P}_k$, $\delta W_k = (\mathcal{L}_{W_k})_* \Sigma_{k+1} - (\mathcal{R}_{W_k})_* \Sigma_k$ and $\delta U_k = U_k \Sigma_k - \mathrm{P}_k U_k$ where $\mathrm{P}_k, \Sigma_k$ are arbitrary with $\mathrm{P}_0, \Sigma_0, \mathrm{P}_N, \Sigma_N$ identically zero.

\item The discrete Euler-Lagrange equations hold:
	\begin{equation}\label{eq:discrete_red_EL}
		\begin{aligned}
			( \mathcal{L}_{V_{k-1}} )^* D_1 \boldsymbol{l}_d( V_{k-1}, U_{k-1}, W_{k-1}) = & (\mathcal{R}_{V_k})^* D_1 \boldsymbol{l}_d( V_k, U_k, W_k )\\
			&+ (\mathcal{R}_{U_k})^* D_2 \boldsymbol{l}_d( V_{k}, U_{k}, W_{k} )\\
			( \mathcal{L}_{W_{k-1}} )^* D_3 \boldsymbol{l}_d( V_{k-1}, U_{k-1}, W_{k-1}) = & (\mathcal{R}_{W_k})^* D_3 \boldsymbol{l}_d( V_k, U_k, W_k )\\
			&- (\mathcal{L}_{U_k})^* D_2 \boldsymbol{l}_d( V_{k}, U_{k}, W_{k})\\
			U_k = &  V_{k-1}^{-1} U_{k-1} W_{k-1}
		\end{aligned}
	\end{equation}
\end{enumerate}
\end{proposition}

\begin{proof}
	The last equation is a consequence of the definitions. The remaining two equations follow from a straightforward computation of the variations and rearrangement of the terms of the sum. See also \cite{MMM06Grupoides} for the general  case of Lie groupoids. 
\end{proof}

These equations are the discrete equivalent of the equations given in theorem \ref{puy}, and under certain regularity conditions they define a discrete flow (see \cite{MMM06Grupoides})
\begin{equation*}
	F_{\boldsymbol{l}_d}( V_{k-1}, U_{k-1}, W_{k-1} ) = ( V_k, U_k, W_{k} ).
\end{equation*}

Much like in the standard setting, in this reduced setting we can define two discrete Legendre transformations 
\begin{equation*}
	\mathcal{F}^{\pm}\boldsymbol{l}_d^{\times} : G \times G \times G \to \mathfrak{g}^* \times G \times \mathfrak{g}^*
\end{equation*}
with coordinate presentation
\begin{align*}
	\mathcal{F}^{+}\boldsymbol{l}_d^{\times}(V_k, U_k, W_k) &= \left( - ( \mathcal{L}_{V_{k}} )^* D_1 \boldsymbol{l}_d( V_{k}, U_{k}, W_{k}),\right.\\
			& \qquad\: \qquad \qquad \qquad V_{k}^{-1} U_{k} W_{k},\\
			& \quad\;\;\: \left.( \mathcal{L}_{W_{k}} )^* D_3 \boldsymbol{l}_d( V_{k}, U_{k}, W_{k})\right)\\
	\mathcal{F}^{-}\boldsymbol{l}_d^{\times}(V_k, U_k, W_k) &= \left( -(\mathcal{R}_{V_k})^* D_1 \boldsymbol{l}_d( V_k, U_k, W_k ) - (\mathcal{R}_{U_k})^* D_2 \boldsymbol{l}_d( V_{k}, U_{k}, W_{k} ), \right.\\
			& \qquad \qquad \qquad \qquad \qquad \qquad \qquad \qquad \qquad \qquad \qquad \quad \; U_k,\\
			& \qquad \left.(\mathcal{R}_{W_k})^* D_3 \boldsymbol{l}_d( V_k, U_k, W_k ) - (\mathcal{L}_{U_k})^* D_2 \boldsymbol{l}_d( V_{k}, U_{k}, W_{k})\right)
\end{align*}

\subsection{The exact discrete Lagrangian}
Given a regular Lagrangian function $\boldsymbol{l}: \mathfrak{g} \times G \times \mathfrak{g} \longrightarrow \mathbb{R}$, we will consider discrete Lagrangians $\boldsymbol{l}_d$ as an approximation to the action of the continuous Lagrangian which can be considered as the exact discrete Lagrangian:
\begin{equation*}
	\boldsymbol{l}_d^e (V_0, U_0, W_0, h) = \int_0^h \boldsymbol{l}(\eta(t), U(t), \psi(t)) \mathrm{d}t,
\end{equation*}
where $t \mapsto (\eta(t), U(t), \psi(t))$ is the unique solution of the Euler-Lagrange equations for $\boldsymbol{l}$
\begin{align*}
	\frac{\mathrm{d}}{\mathrm{d}t}\left( \frac{\partial \boldsymbol{l}}{\partial \eta}\right) &= \mathrm{ad}^*_{\eta} \frac{\partial \boldsymbol{l}}{\partial \eta} - \mathcal{R}_{U}^*\frac{\partial \boldsymbol{l}}{\partial U}\\
	\frac{\mathrm{d}}{\mathrm{d}t}\left( \frac{\partial \boldsymbol{l}}{\partial \psi}\right) &= \mathrm{ad}^*_{\psi} \frac{\partial \boldsymbol{l}}{\partial \psi} + \mathcal{L}_{U}^*\frac{\partial \boldsymbol{l}}{\partial U}\\
	\dot{U} &= U(\psi - \mathrm{Ad}_{U^{-1}}\eta)
\end{align*}
together with the reconstruction equations:
\begin{align*}
	\dot{g}(t) &= g(t) \eta(t)\\
	\dot{\tilde{g}}(t) &= \tilde{g}(t) \psi(t)
\end{align*}
satisfying $g(0) = e$, $g(h) = V_0$, $\tilde{g}(0) = U_0$ and $\tilde{g}(h) = U_0 W_0$ {with small enough $h$} (see \cite{MMM3}). 
%Then for a sufficiently small $h$, the solutions of the DEL equations for $L_d^e$ lie on %the solutions of the Euler-Lagrange equations for $L$, see \cite[Theorem %1.6.4]{marsden-west}.

In practice, it will not be feasible to compute $\boldsymbol{l}_d^e$ and instead we will work with approximations (see \cite{marsden-west}),
\begin{equation*}
	\boldsymbol{l}_d(V_0, U_0, W_0, h) \approx \boldsymbol{l}_d^e(V_0, U_0, W_0, h),
\end{equation*}
of this function. We say that $\boldsymbol{l}_d(V_0, U_0, W_0, h)$ is an \emph{approximation of order} $r$ (to the exact discrete Lagrangian) if there exists an open subset $\mathcal{U}_s \subset G \times \mathfrak{g} \times G \times \mathfrak{g}$ with compact closure and constants $C_s$ and $h_s$ such that
\begin{equation*}
\left\Vert \boldsymbol{l}_d(V(h), U_0, W(h), h) - \boldsymbol{l}_d^e(V(h), U_0, W(h), h) \right\Vert \leq C_s h^{r+1},
\end{equation*}
with $V(h) = g^{-1}(0)g(h)$, $W(h) = \tilde{g}^{-1}(0)\tilde{g}(h)$ and $U_0 = g^{-1}(0)\tilde{g}(0)$, for all solutions $(g(t),\eta(t),\tilde{g}(t),\psi(t))$ of the Euler-Lagrange equations with initial condition in $\mathcal{U}_s$ and for all $h \leq h_s$.

As it is common practice we will fix some $h$ and drop its explicit dependence unless it is strictly necessary. 

In previous sections we consider Lagrangians $\boldsymbol{L}: TG \times TG \to \mathbb{R}$ and their reduced counterparts $\boldsymbol{l}: \mathfrak{g} \times G \times \mathfrak{g} \to \mathbb{R}$. These displayed discrete symmetries of the form $\boldsymbol{L}(\tilde{v}_{\tilde{g}}, v_g) = - \boldsymbol{L}(v_g, \tilde{v}_{\tilde{g}})$ and $\boldsymbol{l}(\eta, U, \psi) = - \boldsymbol{l}(\psi, U^{-1}, \eta)$ and we saw the groupoidal interpretation of this operation on the Hamiltonian side.

In the discrete realm we may define and equivalent transformation $\iota_d: G^4 \to G^4$ and its induced transformation $\check{\iota}_d: G \times G \times G \to G \times G \times G$,
\begin{align*}
	\iota_d(g_1, g_2, \tilde{g}_1, \tilde{g}_2) &= (\tilde{g}_1, \tilde{g}_2, g_1, g_2)\\
	\check{\iota}_d(V, U, W) &= (W, U^{-1}, V)
\end{align*}

We can state the following trivial proposition
\begin{proposition}
Let $\boldsymbol{l}: \mathfrak{g} \times G \times \mathfrak{g} \to \mathbb{R}$ be a Lagrangian satisfying that $\boldsymbol{l}(\psi, U^{-1}, \eta) = - \boldsymbol{l}(\eta, U, \psi)$ for all $(\eta, U, \psi) \in \mathfrak{g} \times G \times \mathfrak{g}$. Then its exact discrete Lagrangian, $\boldsymbol{l}_d^e: G \times G \times G \to \mathbb{R}$, satisfies $\boldsymbol{l}_d^e \circ \check{\iota}_d = - \boldsymbol{l}_d^e$.
\end{proposition}

It is always possible to work with approximations of $\boldsymbol{l}_d^e$ that respect this symmetry, that is, $\boldsymbol{l}_d \approx \boldsymbol{l}_d^e$ satisfying $\boldsymbol{l}_d \circ \check{\iota}_d = - \boldsymbol{l}_d$. Such discrete Lagrangians will be of crucial importance to derive variationally forced integrators.

%Note also that if $\boldsymbol{L}_d: G^4 \to \mathbb{R}$ is a $\widetilde{\Phi}_d$-invariant discrete Lagrangian for the $\widetilde{\Phi}^T$-invariant Lagrangian $\boldsymbol{L}: TG \times TG \to \mathbb{R}$ that satisfies $\boldsymbol{L}_d \circ \iota_d = - \boldsymbol{L}_d$, then the discrete Lagrangian $\boldsymbol{l}_d: G \times G \times G \to \mathbb{R}$ defined by
%\begin{equation}\label{eq:reduced_discrete_lagrangian}
%\boldsymbol{l}_d(V,U,W) = \boldsymbol{L}_d(e, V, U, UW)
%\end{equation}
%is a discrete Lagrangian for $\boldsymbol{l}: \mathfrak{g} \times G \times \mathfrak{g} \to \mathbb{R}$ that satisfies $\boldsymbol{l}_d \circ \check{\iota}_d = - \boldsymbol{l}_d$.

If we define the maps $\epsilon_d: G \times G \rightarrow G^4$ and $\check{\epsilon}_d: G \rightarrow G \times G \times G$ by
\begin{align*}
	\epsilon_d(g_1, g_2) &= (g_1, g_2, g_1, g_2),\\
	\check{\epsilon}_d(V) &= (V, e, V),
\end{align*}
respectively then we can prove the following

\begin{theorem}\label{thm:discrete_flow_identities}
The discrete flow $F_{\boldsymbol{l}_d}: G \times G \times G \rightarrow G \times G \times G$ defined by a discrete Lagrangian $\boldsymbol{l}_d: G \times G \times G \rightarrow \mathbb{R}$ verifying that $\boldsymbol{l}_d \circ \check{\iota}_d = -\boldsymbol{l}_d$ restricts to $\check{\epsilon}_d(G)$, that is, 
	\begin{equation*}
		F_{\boldsymbol{l}_d} \circ \check{\epsilon}_d (G) \in \check{\epsilon}_d (G).
	\end{equation*}
\end{theorem}
\begin{proof}
If we apply the identity $\boldsymbol{l}_d \circ \check{\iota}_d = - \boldsymbol{l}_d$ to
	\begin{equation}\label{eq:discret_red_action_inv}
		\sum_{i = 0}^N \left(\boldsymbol{l}_d \circ \check{\iota}_d\right)(V_k, U_k, W_k)
	\end{equation}
and apply the discrete Hamilton principle we obtain eq. \eqref{eq:discret_red_action}, and it follows immediately that solutions of the system
	\begin{align*}
		( \mathcal{L}_{W_{k-1}} )^* D_1 \boldsymbol{l}_d( W_{k-1}, U^{-1}_{k-1}, V_{k-1}) &= (\mathcal{R}_{W_k})^* D_1 \boldsymbol{l}_d( W_k, U^{-1}_k, V_k )\\
		&+ (\mathcal{R}_{U^{-1}_k})^* D_2 \boldsymbol{l}_d( W_{k}, U^{-1}_{k}, V_{k} ),\\
		( \mathcal{L}_{V_{k-1}} )^* D_3 \boldsymbol{l}_d( W_{k-1}, U^{-1}_{k-1}, V_{k-1}) &= (\mathcal{R}_{V_k})^* D_3 \boldsymbol{l}_d( W_k, U^{-1}_k, V_k )\\
		&- (\mathcal{L}_{U^{-1}_k})^* D_2 \boldsymbol{l}_d( W_{k}, U^{-1}_{k}, V_{k}),\\
		U^{-1}_k &= W_{k-1}^{-1} U^{-1}_{k-1} V_{k-1},
	\end{align*}
obtained from varying eq. \eqref{eq:discret_red_action_inv} must also be solutions of eqs. \eqref{eq:discrete_red_EL}. If we restrict either these or eqs. \eqref{eq:discrete_red_EL} to $\check{\epsilon}_d(G)$ the last equation turns into an identity and the remaining equations become
	\begin{align*}
		( \mathcal{L}_{V_{k-1}} )^* D_1 \boldsymbol{l}_d( V_{k-1}, e, V_{k-1}) &= (\mathcal{R}_{V_k})^* D_1 \boldsymbol{l}_d( V_k, e, V_k ) + D_2 \boldsymbol{l}_d( V_{k}, e, V_{k} ),\\
		( \mathcal{L}_{V_{k-1}} )^* D_3 \boldsymbol{l}_d( V_{k-1}, e, V_{k-1}) &= (\mathcal{R}_{V_k})^* D_3 \boldsymbol{l}_d( V_k, e, V_k ) - D_2 \boldsymbol{l}_d( V_{k}, e, V_{k}).
	\end{align*}
The vanishing of the dynamics in $U_k$ proves that $\left.F_{\boldsymbol{l}_d}\right\vert_{\check{\epsilon}_d(G)}: \check{\epsilon}_d(G) \rightarrow \check{\epsilon}_d(G)$.
\end{proof}

\begin{theorem}\label{main-theorem-1}
Let $(L, F)$ be a $\widehat{\Phi}^{\times}$-invariant forced regular Lagrangian system in $G$ such that it defines an $(l, f)$ forced regular Lagrangian system in $\mathfrak{g}$. Denote by $\boldsymbol{L}_{F}: TG \times TG \to \mathbb{R}$ the extended Lagrangian, and let $\boldsymbol{L}_{F,d}: G^4 \to \mathbb{R}$ be a $\widehat{\Phi}_d$-invariant approximation to the exact discrete Lagrangian of order $r$ satisfying $\boldsymbol{L}_{F,d} \circ \iota_d = -\boldsymbol{L}_{F,d}$. Then,
\begin{itemize}
	\item The discrete Lagrangian $\boldsymbol{l}_{f,d}: G \times G \times G \to \mathbb{R}$ defined by
		\begin{equation*}
			\boldsymbol{l}_{f,d}(V,U,W) = \boldsymbol{L}_{F,d}(e, V, U, UW),
		\end{equation*}
	is an approximation of order $r$ for $\boldsymbol{l}_{f,d}^e$ satisfying the identity $\boldsymbol{l}_{f,d} \circ \check{\iota}_d = -\boldsymbol{l}_{f,d}$.
	\item When restricted to $\epsilon_d(G \times G)$, the discrete flow $F_{\boldsymbol{L}_{F,d}}: G^4 \to G^4$ induced by its discrete Euler-Lagrange equations is an approximation of order $r$ to the flow of $(L,F)$.
	\item When restricted to $\check{\epsilon}_d(G)$, the discrete flow $F_{\boldsymbol{l}_{f,d}}: G \times G \times G \to G \times G \times G$ induced by its discrete Euler-Lagrange equations is an approximation of order $r$ to the flow of $(l,f)$.
\end{itemize}
\end{theorem}

\begin{proof}
We have that for $h$ sufficiently small
\begin{align*}
	\boldsymbol{L}_{F,d}(g_0, g_1, \tilde{g}_0, \tilde{g}_1) &= \boldsymbol{L}_{F,d}^e(g_0, g_1, \tilde{g}_0, \tilde{g}_1) + \mathcal{O}\left(h^{r+1}\right)\\
	&= \int_0^h \boldsymbol{L}_F(g(t),\dot{g}(t),\tilde{g}(t),\dot{\tilde{g}}(t)) \mathrm{d}t + \mathcal{O}\left(h^{r+1}\right).
\end{align*}
where $(g(t),\dot{g}(t),\tilde{g}(t),\dot{\tilde{g}}(t))$ is a solution of the Euler-Lagrange equations for $\boldsymbol{L}_F$ such that
$g(0) = g_0 = e$, $g(h) = g_1$, $\tilde{g}(0) = \tilde{g}_0$, $\tilde{g}(h) = \tilde{g}_1$.
This means that, by $\widehat{\Phi}^{\times}$-invariance,
\begin{equation*}
	\boldsymbol{L}_{F,d}(g_0, g_1, \tilde{g}_0, \tilde{g}_1) = \int_0^h \boldsymbol{l}_F(\eta(t), U(t), \psi(t)) \mathrm{d}t + \mathcal{O}\left(h^{r+1}\right),
\end{equation*}
with $\eta(t) = g^{-1}(t)\dot{g}(t)$, $U(t) = g^{-1}(t)\tilde{g}(t)$ and $\psi(t) = \tilde{g}^{-1}(t)\dot{\tilde{g}}(t)$ and by $\widehat{\Phi}_d$-invariance we get
\begin{equation*}
\boldsymbol{L}_{F,d}(g_0, g_1, \tilde{g}_0, \tilde{g}_1) = \boldsymbol{L}_{F,d}(e, g_0^{-1} g_1, g_0^{-1} \tilde{g}_0, g_0^{-1} \tilde{g}_1) = \boldsymbol{l}_{f,d}(V_0, U_0, W_0),
\end{equation*}
with $V_0 = g_0^{-1} g_1$, $U_0 = g_0^{-1} \tilde{g}_0$, $W_0 = \tilde{g}_0^{-1} \tilde{g}_1$. Thus $\boldsymbol{l}_{f,d}$ is indeed of order $r$ with respect to $\boldsymbol{l}_{f,d}^e$.

That $\boldsymbol{l}_{f,d}$ satisfies $\boldsymbol{l}_{f,d} \circ \check{\iota}_d = -\boldsymbol{l}_{f,d}$ follows immediately from its definition from $\boldsymbol{L}_{F,d}$, i.e.,
\begin{align*}
(\boldsymbol{l}_{f,d} \circ \check{\iota}_d)(V_0, U_0, W_0) &= (\boldsymbol{L}_{F,d} \circ \iota_d)(g_0, g_1, \tilde{g}_0, \tilde{g}_1)\\
&= -\boldsymbol{L}_{F,d}(\tilde{g}_0, \tilde{g}_1, g_0, g_1),\\
&= -\boldsymbol{L}_{F,d}(e, \tilde{g}_0^{-1} \tilde{g}_1, \tilde{g}_0^{-1} g_0, \tilde{g}_0^{-1} g_1),\\
&= -\boldsymbol{l}_{f,d}(W_0, U_0^{-1},V_0).
\end{align*}

In the second point it suffices to apply the variational error theorem from \cite{PatrickCuell}, which proves that $F_{\boldsymbol{L}_{F,d}}$ is an approximation of order $r$ to the exact Hamiltonian flow induced by the Euler-Lagrange equations. Afterwards, we need only to apply theorem \ref{thm:discrete_flow_identities} to see that the discrete flow projects onto $G$, thus approximating the continuous flow for the forced Lagrangian system $(L,F)$.

The third point can then be seen as a direct consequence of the second point. If $\tilde{\pi}_d: G^4 \to G^4/ G \equiv G \times G \times G$, then it is clear that $\tilde{\pi}_d \circ F_{\boldsymbol{L}_{F,d}} = F_{\boldsymbol{l}_{f,d}} \circ \tilde{\pi}_d$, and $\tilde{\pi}_d$ does not affect the order, so the result follows immediately.
\end{proof}

\begin{remark}
This theorem can be proven without mentioning the forced system $(L, F)$ or the discrete Lagrangian $\boldsymbol{L}_{F,d}$ at all, by directly applying the results of \cite{MMM3} (theorem 5.7 in particular) and then applying theorem \ref{thm:discrete_flow_identities}.
\end{remark}

\subsection{Variationally partitioned Runge-Kutta-Munthe-Kaas methods with forcing}
There exist several ways to construct high-order variational integrators. We will be focusing here on those based on Runge-Kutta (RK) methods. These are a well-known type of numerical method used to approximate the solution of differential equations. A fixed-step RK method (for an autonomous system) is characterized by a series of coefficients $(a_{i j}, b_j)$, $1 < i,j < s$, and, given an initial value problem $\dot{y}(t) = f(y(t))$, $y(t_0) = y_0 \in \mathbb{R}^n$, the generic form of a step of size $h$ is
\begin{align*}
Y_k^i &= y_k + h \sum_{j = 1}^s a_{i j} f(Y_k^j),\\
y_{k+1} &= y_k + h \sum_{j = 1}^s b_j f(Y_k^i).
\end{align*}

Constructing of variational integrators using RK methods means using these methods to \emph{discretize} the equations relating our position variables with the velocities, i.e. $\dot{q}(t) = v(t)$. In the case of a Lie group these equations constitute the reconstruction equations, which in matrix notation take the form $\dot{g}(t) = g(t)\eta(t)$. As we are working on a manifold one must be careful to bring all operations to common vector spaces where the methods can be applied. The resulting methods using this strategy receive the name of Runge-Kutta-Munthe-Kaas methods (RKMK) \cite{Lie-group}

Consider a retraction map given by a local diffeomorphism $\tau: \mathfrak{g} \to \mathcal{U}_e \subset G$, where $\mathcal{U}_e$ is a neighbourhood of the identity element. The most common instances of these are the exponential map, $\exp$, and the Cayley map, $\mathrm{cay}$ (in the case of quadratic Lie groups) \cite{iserles}.

Assuming that $G$ is connected, we will be able to transport a neighbourhood of any point to  $\mathcal{U}_e$ and from $\mathcal{U}_e$ to $\mathfrak{g}$ and back thanks to $\tau^{-1}$ and $\tau$. Not only that, but this will be also possible in $\mathbb{T}G = TG \oplus T^*G$, which is what we need for our mechanical problems.

For some $h \in G$, the complete geometric scheme is as follows:
\begin{center}
	\begin{tikzcd}[column sep=huge,row sep=huge]
		\mathbb{T}\mathfrak{g} \arrow[d, "\pi_\mathfrak{g}"]  \arrow[r, shift left=1mm, "\mathbb{T}\tau"] \arrow[r, shift right=1mm, leftarrow, "\mathbb{T}\tau^{-1}"'] & \mathbb{T}\mathcal{U}_e \arrow[d, "\pi_U"] \arrow[r, shift left=1mm, "\mathbb{T}\mathcal{L}_{h}"] \arrow[r, shift right=1mm, leftarrow, "\mathbb{T}\mathcal{L}_{h}^{-1} = \mathbb{T}\mathcal{L}_{h^{-1}}"'] & \mathbb{T}G \arrow[d, "\pi_G"]\\
		\mathfrak{g} \arrow[r, shift left=1mm, "\tau"] \arrow[r, shift right=1mm, leftarrow, "\tau^{-1}"'] & \mathcal{U}_e \arrow[r, shift left=1mm, "\mathcal{L}_{h}"] \arrow[r, shift right=1mm, leftarrow, "\mathcal{L}_{h}^{-1} = \mathcal{L}_{h^{-1}}"'] & G
	\end{tikzcd}
\end{center}
with 
\begin{eqnarray*}
	\mathbb{T}\tau(\xi, \eta, \mu)&=&(\tau(\xi), T_{\xi}\tau(\xi),(T_{\tau(\xi)}\tau^{-1})^* \mu)
	\end{eqnarray*}
where $\xi \in  \mathfrak{g}$, $\eta \in T_\xi \mathfrak{g} \cong \mathfrak{g}$ and $\mu \in T^*_\xi \mathfrak{g} \cong \mathfrak{g}^*$  and similar definitions for the other maps.

Assume we work in adapted coordinates $(g, v, p) \in \mathbb{T}G$ and $(\xi, \eta, \mu) \in \mathbb{T} \mathfrak{g}$. According to this diagram, if $h$ is such that $\mathcal{L}_{h^{-1}} g \in \mathcal{U}_e$, we find the following correspondences:
\begin{align*}
&\mathbb{T}_{\mathcal{L}_{h^{-1}} g} \tau^{-1} (\mathbb{T}_{g}\mathcal{L}_{h^{-1}}(g,v,p))\\
&\quad = \left( \tau^{-1}\left(\mathcal{L}_{h^{-1}} g\right), \mathrm{d}^L \tau^{-1}_{\tau^{-1} (\mathcal{L}_{h^{-1}} g)} (T_{g} \mathcal{L}_{g^{-1}} v),  \left(\mathrm{d}^L \tau_{\tau^{-1}(\mathcal{L}_{h^{-1}} g)}\right)^* \left(T_{e} \mathcal{L}_{g}\right)^* p\right)\\
&\quad = \left( \xi, \eta, \mu \right)\\
&\mathbb{T}_{\tau(\xi)}(\mathcal{L}_{h} \mathbb{T}_{\xi} \tau (\xi,\eta, \mu))\\
&\quad = \left( \mathcal{L}_h \tau(\xi), T_e \mathcal{L}_{\mathcal{L}_h \tau(\xi)} (\mathrm{d}^L \tau_{\xi} \eta), \left(T_{\mathcal{L}_h \tau(\xi)} \mathcal{L}_{\left(\mathcal{L}_h \tau(\xi)\right)^{-1}}\right)^* \left(\mathrm{d}^L \tau^{-1}_{\xi}\right)^* \mu \right)\\
&\quad = \left( g, v, p \right)
\end{align*}
where $\mathrm{d}^{L}\tau: \mathfrak{g} \times \mathfrak{g} \to \mathfrak{g}$ is the left-trivialised tangent to $\tau$ defined by the relation $(T_\xi \tau) \eta_{\xi} = T_e \mathcal{L}_{\tau(\xi)} \left(\mathrm{d}^L \tau_{\xi} \eta_{\xi}\right)$, for $\eta_{\xi} \in T\mathfrak{g}$, and $\mathrm{d}^{L}\tau^{-1}$ is its inverse, defined by $(T_g \tau^{-1}) v_{g} = \mathrm{d}^L \tau_{\tau^{-1}({g})}^{-1} \left(T_g \mathcal{L}_{g^{-1}} v_{g}\right)$, for $v_{g} \in TG$ (see \cite{MR2496560}).

Let us also take this opportunity to define $\mathrm{dd}^{L}\tau: \mathfrak{g} \times \mathfrak{g} \times \mathfrak{g} \to \mathfrak{g}$ the second left-trivialised tangent, which will be necessary for later derivations. This is a linear map in the second and third variables such that $\partial_{\xi} \left(\mathrm{d}^{L}\tau_{\xi} \eta\right) \delta \xi = \mathrm{d}^{L}\tau_{\xi} \mathrm{dd}^{L}\tau_{\xi}(\eta, \delta \xi)$. It appears naturally when representing elements $(g, v, a) \in T^{(2)}G$, the second order tangent bundle of $G$ (see \cite{ManuelCampos}), with elements of $(\xi, \eta, \zeta) \in T^{(2)}\mathfrak{g}$, which using matrix notation becomes:
\begin{equation*}
\left(e, g^{-1} v, g^{-1} a - g^{-1} v g^{-1} v \right) \mapsto \left(0, \mathrm{d}^L\tau_{\xi}\eta, \mathrm{d}^L\tau_{\xi}\left[\zeta +  \mathrm{dd}^L\tau_{\xi}\left(\eta,\eta\right)\right] \right)
\end{equation*}

With this we can now proceed to propose a discretization of the reconstruction equations:
\begin{align*}
\tau^{-1}((g_k)^{-1} G_k^i) &= h \sum_{j=1}^s a_{i j} \mathrm{d}^{L}\tau^{-1}_{\tau^{-1}((g_k)^{-1} G_k^j)} (G_k^j)^{-1} V_k^j,\\
\tau^{-1}((g_k)^{-1} g_{k+1}) &= h \sum_{j=1}^s b_{j} \mathrm{d}^{L}\tau^{-1}_{\tau^{-1}((g_k)^{-1} G_k^j)} (G_k^j)^{-1} V_k^j.
\end{align*}
which, if $\Xi_k^i = \tau^{-1}((g_k)^{-1} G_k^i)$, $\xi_{k, k+1} = \tau^{-1}((g_k)^{-1} g_{k+1})$ and $\left(G_k^i\right)^{-1} V_k^i = \mathrm{d}^{L}\tau_{\Xi_k^i} \mathrm{H}_k^i$ (where $\mathrm{\Xi}$ and $\mathrm{H}$ correspond to the variables $\xi$ and $\eta$ respectively), reduces to
\begin{align*}
\Xi_k^i &= h \sum_{j=1}^s a_{i j} \mathrm{H}_k^j,\\
\xi_{k, k+1} &= h \sum_{j=1}^s b_{j} \mathrm{H}_k^j.
\end{align*}

If we consider the pair groupoid $TG \times TG$ with local coordinates $(g,\dot{g},\widetilde{g},\dot{\widetilde{g}})$, a regular Lagrangian $\boldsymbol{L}: TG \times TG \to \mathbb{R}$, and a quadrature rule associated to the RK method we want to apply, an approximation to the exact discrete Lagrangian can be written as
\begin{equation*}
\boldsymbol{L}_d (g_0,g_N,\widetilde{g}_0,\widetilde{g}_N) = \sum_{k = 0}^{N-1} \boldsymbol{L}_d (g_k,g_{k+1},\widetilde{g}_k,\widetilde{g}_{k+1}) = \sum_{k = 0}^{N-1} \sum_{i = 1}^{s} b_i h \boldsymbol{L}(G^i_k, \dot{G}^i_k, \widetilde{G}^i_k, \dot{\widetilde{G}}^i_k)
\end{equation*}
where
\begin{align*}
G^i_k &= g_k \tau\left(\Xi_k^i\right),\\
\dot{G}^i_k &= g_k \tau\left(\Xi_k^i\right) \mathrm{d}^{L} \tau_{\Xi_k^i} \mathrm{H}^i_k,\\
\widetilde{G}^i_k &= \tilde{g}_k \tau\left(\mathrm{X}_k^i\right),\\
\dot{\widetilde{G}}^i_k &= \tilde{g}_k \tau\left(\mathrm{X}_k^i\right) \mathrm{d}^{L} \tau_{\mathrm{X}_k^i} \Psi^i_k,\\
\end{align*}
and $(\Xi_k^i, \mathrm{H}_k^i, \mathrm{X}_k^i, \Psi^i_k) \in T\mathfrak{g} \times T\mathfrak{g}$ are chosen so as to extremize the discrete action subject to the constraints
\begin{align*}
\Xi_k^i &= h \sum_{j=1}^s a_{i j} \mathrm{H}_k^j,\\
\tau^{-1}((g_k)^{-1} g_{k+1}) = \xi_{k, {k+1}} &= h \sum_{j=1}^s b_{j} \mathrm{H}_k^j\\
\mathrm{X}_k^i &= h \sum_{j=1}^s a_{i j} \Psi_k^j,\\
\tau^{-1}((\widetilde{g}_k)^{-1} \widetilde{g}_{k+1}) = \chi_{k, {k+1}} &= h \sum_{j=1}^s b_{j} \Psi_k^j.
\end{align*}

If $\boldsymbol{L}$ is $\widehat{\Phi}^{\times}$-invariant, then the discrete Lagrangian can be rewritten as
\begin{align*}
\boldsymbol{L}_d (g_k,g_{k+1},\widetilde{g}_k,\widetilde{g}_{k+1}) &= h \sum_{i = 1}^{s} b_i \boldsymbol{l}(\mathrm{d}^L \tau_{\Xi_k^i} \mathrm{H}_k^i, \tau(-\Xi_k^i) g_k^{-1} \widetilde{g}_k \tau(\mathrm{X}_k^i), \mathrm{d}^L \tau_{\mathrm{X}_k^i} \Psi_k^i)\\
&= \boldsymbol{l}_d (V_k \equiv g_k^{-1} g_{k+1}, U_k \equiv g_k^{-1} \widetilde{g}_k, W_k \equiv \widetilde{g}_k^{-1} \widetilde{g}_{k+1}).
\end{align*}

The equations resulting from the extremization process are
\begin{align*}
\lambda_{k+1} &= \mathrm{Ad}^*_{\xi_{k,k+1}} \left[\lambda_k + h \sum_{j=1}^s b_j \mathcal{R}^*_{U_k} \widehat{\mathrm{K}}_k^i\right]\\
\mu_{k+1} &= \mathrm{Ad}^*_{\chi_{k,k+1}} \left[\mu_k + h \sum_{j=1}^s b_j \mathcal{L}^*_{U_k} \widehat{\mathrm{K}}_k^i\right]\\
\Lambda_{k}^i &= \mathrm{Ad}^*_{\xi_{k,k+1}} \left[\lambda_k + h \sum_{j=1}^s b_j \left( \mathcal{R}^*_{U_k} \widehat{\mathrm{K}}_k^j - \frac{a_{j i}}{b_i} \left(\mathrm{d}^L \tau_{\Xi_k^j} \mathrm{d}^L \tau_{-\xi_{k,k+1}}^{-1}\right)^* \mathcal{R}^*_{\tau(-\Xi^j_k) U_k \tau(\mathrm{X}^j_k)} \mathrm{K}_k^j \right)\right]\\
\mathrm{M}_{k}^i &= \mathrm{Ad}^*_{\chi_{k,k+1}} \left[\mu_k + h \sum_{j=1}^s b_j \left( \mathcal{L}^*_{U_k} \widehat{\mathrm{K}}_k^j - \frac{a_{j i}}{b_i} \left(\mathrm{d}^L \tau_{\mathrm{X}_k^j} \mathrm{d}^L \tau^{-1}_{-\chi_{k,k+1}}\right)^* \mathcal{L}^*_{\tau(-\Xi^j_k) U_k \tau(\mathrm{X}^j_k)} \mathrm{K}_k^j\right)\right]
\end{align*}
where
\begin{align*}
\lambda_{k} &= - D_1 \boldsymbol{l}(\eta_k, U_k, \psi_k)\\
\mu_{k} &= D_3 \boldsymbol{l}(\eta_k, U_k, \psi_k)\\
\Pi_{k}^i &= - \left(\mathrm{d}^L \tau_{\Xi_{k}^i}\right)^* D_1 \boldsymbol{l}(\mathrm{d}^L \tau_{\Xi_k^i} \mathrm{H}_k^i, \tau(-\Xi_k^i) U_k \tau(\mathrm{X}_k^i), \mathrm{d}^L \tau_{\mathrm{X}_k^i} \Psi_k^i))\\
\mathrm{P}_{k}^i &= \left(\mathrm{d}^L \tau_{\mathrm{X}_{k}^i}\right)^* D_3 \boldsymbol{l}(\mathrm{d}^L \tau_{\Xi_k^i} \mathrm{H}_k^i, \tau(-\Xi_k^i) U_k \tau(\mathrm{X}_k^i), \mathrm{d}^L \tau_{\mathrm{X}_k^i} \Psi_k^i))\\
\Lambda_k^i &=  \left(\mathrm{d}^L \tau^{-1}_{\xi_{k,k+1}}\right)^* \left[ \Pi_k^i  + h \sum_{j = 1}^s \frac{b_j a_{j i}}{b_i} \left(\mathrm{dd}^L \tau^{-1}_{\Xi_{k}^j}\right)^* (\mathrm{H}^j_k, \Pi_k^j)\right]\\
\mathrm{M}_k^i &=  \left(\mathrm{d}^L \tau^{-1}_{\chi_{k,k+1}}\right)^* \left[ \mathrm{P}_k^i  + h \sum_{j = 1}^s \frac{b_j a_{j i}}{b_i} \left(\mathrm{dd}^L \tau^{-1}_{\mathrm{X}_{k}^j}\right)^* (\Psi^j_k, \mathrm{P}_k^j)\right]\\
\mathrm{K}_k^i &= D_2 \boldsymbol{l}(\mathrm{d}^L \tau_{\Xi_k^i} \mathrm{H}_k^i, \tau(-\Xi_k^i) U_k \tau(\mathrm{X}_k^i), \mathrm{d}^L \tau_{\mathrm{X}_k^i} \Psi_k^i))\\
\widehat{\mathrm{K}}_k^i &= \mathcal{L}^*_{\tau(-\Xi^i_k)} \mathcal{R}^*_{\tau(\mathrm{X}^i_k)} \mathrm{K}_k^i\,.
\end{align*}

Note that $V_k = \tau(\xi_{k,k+1})$ and $W_k = \tau(\chi_{k,k+1})$, so we can write
\begin{equation*}
U_{k+1} = \tau(-\xi_{k,k+1}) U_k \tau(\chi_{k,k+1}).
\end{equation*}

Restriction to the identities in this setting means $\Xi_k^i = \mathrm{X}_k^i$, $\xi_{k,k+1} = \chi_{k,k+1}$, $\mathrm{H}_k^i = \Psi_k^i$, $\lambda_{k} = \mu_k$, $\Lambda_{k}^i = \mathrm{M}_k^i$ and $U_k = e$. For a Lagrangian $\boldsymbol{l}_f$, when we restrict these equations to the identities we find that
\begin{align*}
\lambda_{k} &= l'(\eta_k)\\
\Pi_{k}^i &= \left(\mathrm{d}^L \tau_{\Xi_{k}^i}\right)^* l'(\mathrm{d}^L \tau_{\Xi_k^i} \mathrm{H}_k^i)\\
\mathrm{K}_k^i &= f(\mathrm{d}^L \tau_{\Xi_k^i} \mathrm{H}_k^i)\\
\widehat{\mathrm{K}}_k^i &= \left(\mathrm{Ad}_{\tau(\Xi^i_k)}^{-1}\right)^* f(\mathrm{d}^L \tau_{\Xi_k^i} \mathrm{H}_k^i).
\end{align*}

It is convenient to define
\begin{equation*}
\mathrm{N}_k^i = \left(\mathrm{d}^L \tau_{\Xi_{k}^i}\right)^* f(\mathrm{d}^L \tau_{\Xi_k^i} \mathrm{H}_k^i),
\end{equation*}
which, taking into account that $\mathrm{Ad}_{\tau(\xi)}^{-1} = \mathrm{d}^L \tau_{\xi} \mathrm{d}^L \tau^{-1}_{-\xi}$, allows us to write
\begin{equation*}
\widehat{\mathrm{K}}_k^i = \left(\mathrm{d}^L \tau^{-1}_{-\Xi_{k}^i}\right)^* \mathrm{N}_k^i.
\end{equation*}
This lets us rewrite the equations at the identities as
\begin{align*}
\lambda_{k+1} &= \mathrm{Ad}^*_{\xi_{k,k+1}} \left[\lambda_k + h \sum_{j=1}^s b_j \left(\mathrm{d}^L \tau^{-1}_{-\Xi_{k}^j}\right)^* \mathrm{N}_k^j\right]\\
\Lambda_{k}^i &= \mathrm{Ad}^*_{\xi_{k,k+1}} \left[\lambda_k + h \sum_{j=1}^s b_j \left(\mathrm{d}^L \tau^{-1}_{-\Xi_{k}^j} - \frac{a_{j i}}{b_i} \mathrm{d}^L \tau_{-\xi_{k,k+1}}^{-1} \right)^* \mathrm{N}_k^j \right],
\end{align*}
which is precisely the form the variationally partitioned RKMK equations were expected to take for a reduced Lagrangian with forcing.

\subsection{Numerical tests}
\subsubsection{Simplified Landau-Lifschitz-Gilbert model}\label{subsection1}
For our first set of numerical tests we have chosen a simplified version of the Landau-Lifschitz-Gilbert (LLG) model for ferromagnetic materials (see \cite{LandauCollected}).

The configuration manifold of the system is the Lie group $SO(3)$, whose Lie algebra is $\mathfrak{so}(3) \cong \mathbb{R}^3$. Its velocity phase space is therefore $TSO(3) \equiv SO(3) \times \mathbb{R}^3$. Its Lagrangian $L: TSO(3) \to \mathbb{R}$ is just the standard rigid body Lagrangian, which is invariant under the action of the group; therefore, we may work with the following reduced Lagrangian $\ell: \mathbb{R}^3 \to \mathbb{R}$
\begin{equation*}
\ell(\mathbf{\Omega}) = \frac{1}{2} \mathbf{\Omega}^T I \mathbf{\Omega}
\end{equation*}
where $I$ denotes here the inertia tensor of the particle and $\mathbf{\Omega} \in \mathbb{R}^3$, with coordinates $(\Omega_x, \Omega_y, \Omega_z)$.

The Euler-Poincar{\'e} equations for this simple Lagrangian are the well-known Euler equations for the rigid body,
\begin{equation*}
\dot{\mathbf{M}} = \mathbf{M} \times \mathbf{\Omega}
\end{equation*}
where $\frac{\partial \ell}{\partial \mathbf{\Omega}} = \mathbf{M} = I \mathbf{\Omega}$ and $\mathrm{ad}_{\mathbf{\Omega}}^* \frac{\partial \ell}{\partial \mathbf{\Omega}} = \mathbf{M} \times \mathbf{\Omega}$.

The simplified LLG force is of the form
\begin{equation*}
f = \alpha \mathbf{M} \times (\mathbf{M} \times \mathbf{\Omega})
\end{equation*}
with $\alpha \in \mathbb{R}$ a constant. Therefore, the equations of motion are
\begin{equation*}
\dot{\mathbf{M}} = \mathbf{M} \times \mathbf{\Omega} + \alpha \mathbf{M} \times (\mathbf{M} \times \mathbf{\Omega})\,.
\end{equation*}

This is a simple model for so-called \emph{double bracket dissipation} \cite{BKMR}, which is known to preserve Casimir functions such as
\begin{equation*}
C = \mathbf{M}^2 = (I_x \Omega_x)^2 + (I_y \Omega_y)^2 + (I_z \Omega_z)^2\,.
\end{equation*}
The integrator does not preserve this function exactly, being a general quadratic invariant, although it seems to be preserved in the free case.

\begin{figure}[!ht]
  \centering
  \begin{subfigure}[b]{0.45\textwidth}
  	\includegraphics[scale=0.50, clip=true, trim=33mm 90mm 44mm 90mm]{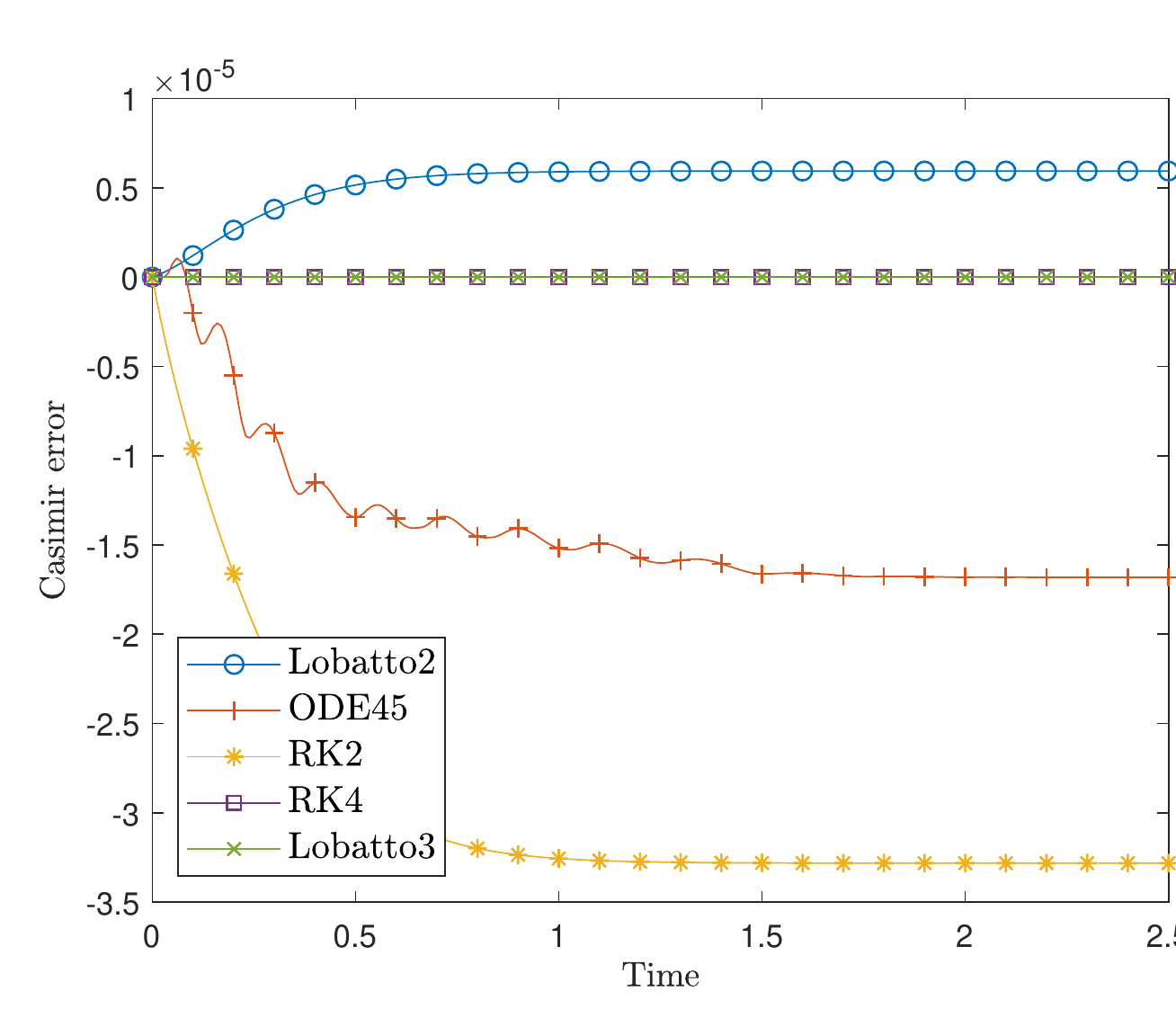}
  \end{subfigure}
  \hspace{0.5cm}
  \begin{subfigure}[b]{0.45\textwidth}
  	\includegraphics[scale=0.50, clip=true, trim=33mm 90mm 42mm 90mm]{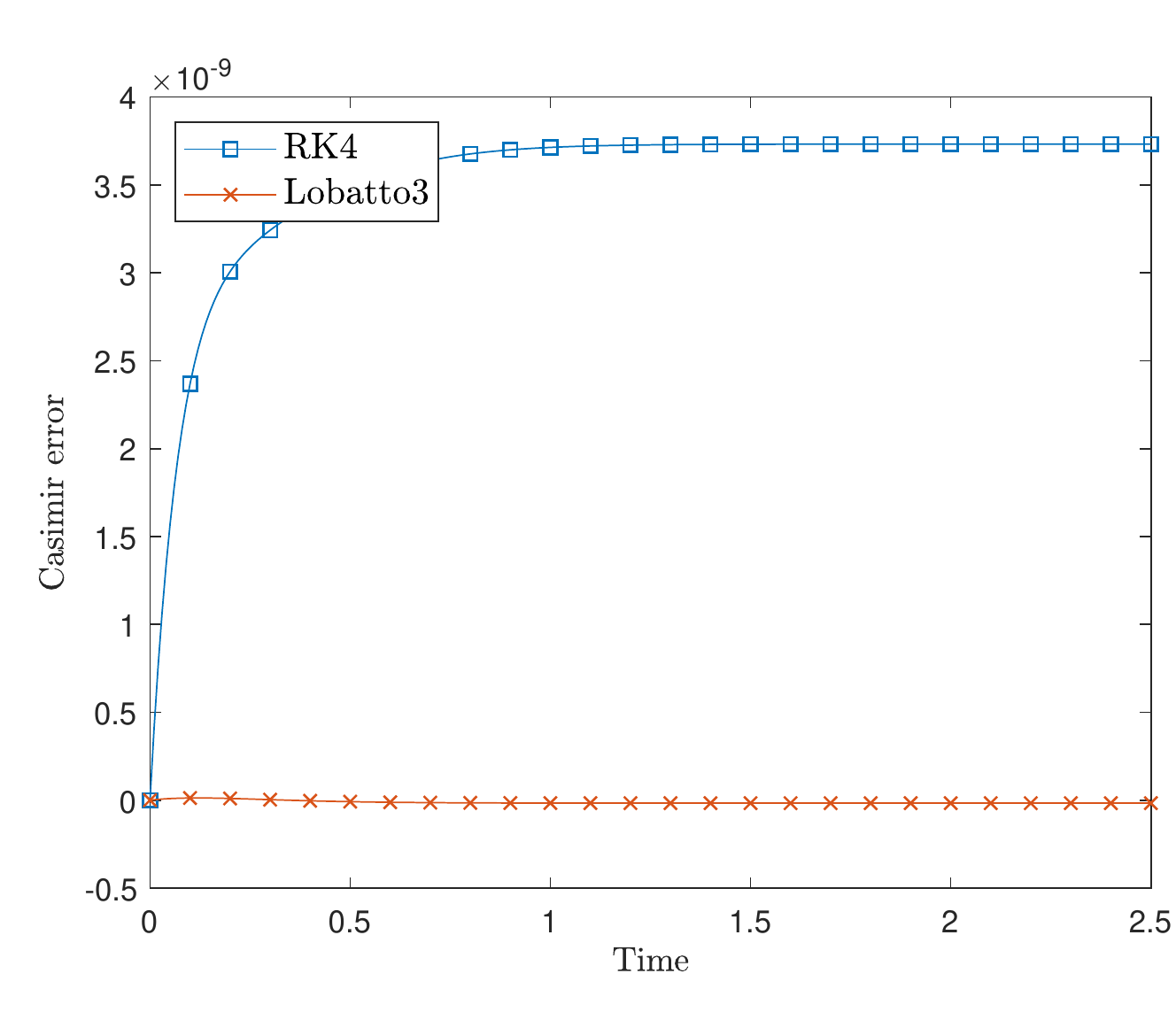}
  \end{subfigure}
  \caption{Casimir evolution. The difference between the initial value and its value at any other moment is displayed. On the left we have a battery of different integrators of different orders. On the right we have two constant-step order 4 integrators, the classic Runge-Kutta 4 and a 3-stage partitioned Lobatto method. All constant-step methods use $h = 0.01$. Note that ODE45 is a variable-step method of order 4 with an order 5 estimator, and in this case it behaves rather poorly. }
  \label{fig:LLG_casimir_plots}
\end{figure}

\begin{figure}[!ht]
  \centering
  \begin{subfigure}[b]{0.45\textwidth}
  	\includegraphics[scale=0.50, clip=true, trim=33mm 90mm 44mm 90mm]{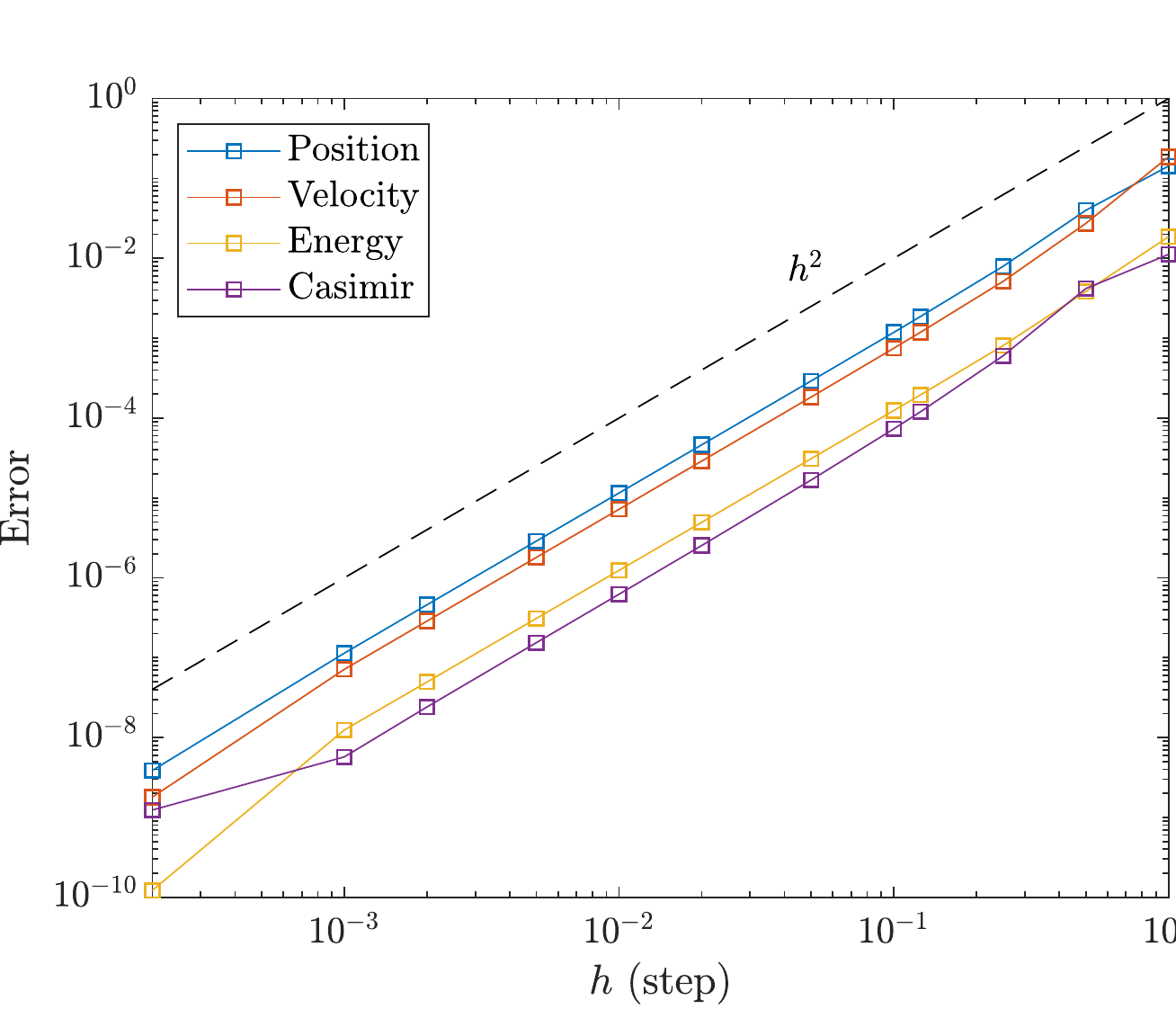}
  \end{subfigure}
  \hspace{0.5cm}
  \begin{subfigure}[b]{0.45\textwidth}
  	\includegraphics[scale=0.50, clip=true, trim=33mm 90mm 42mm 90mm]{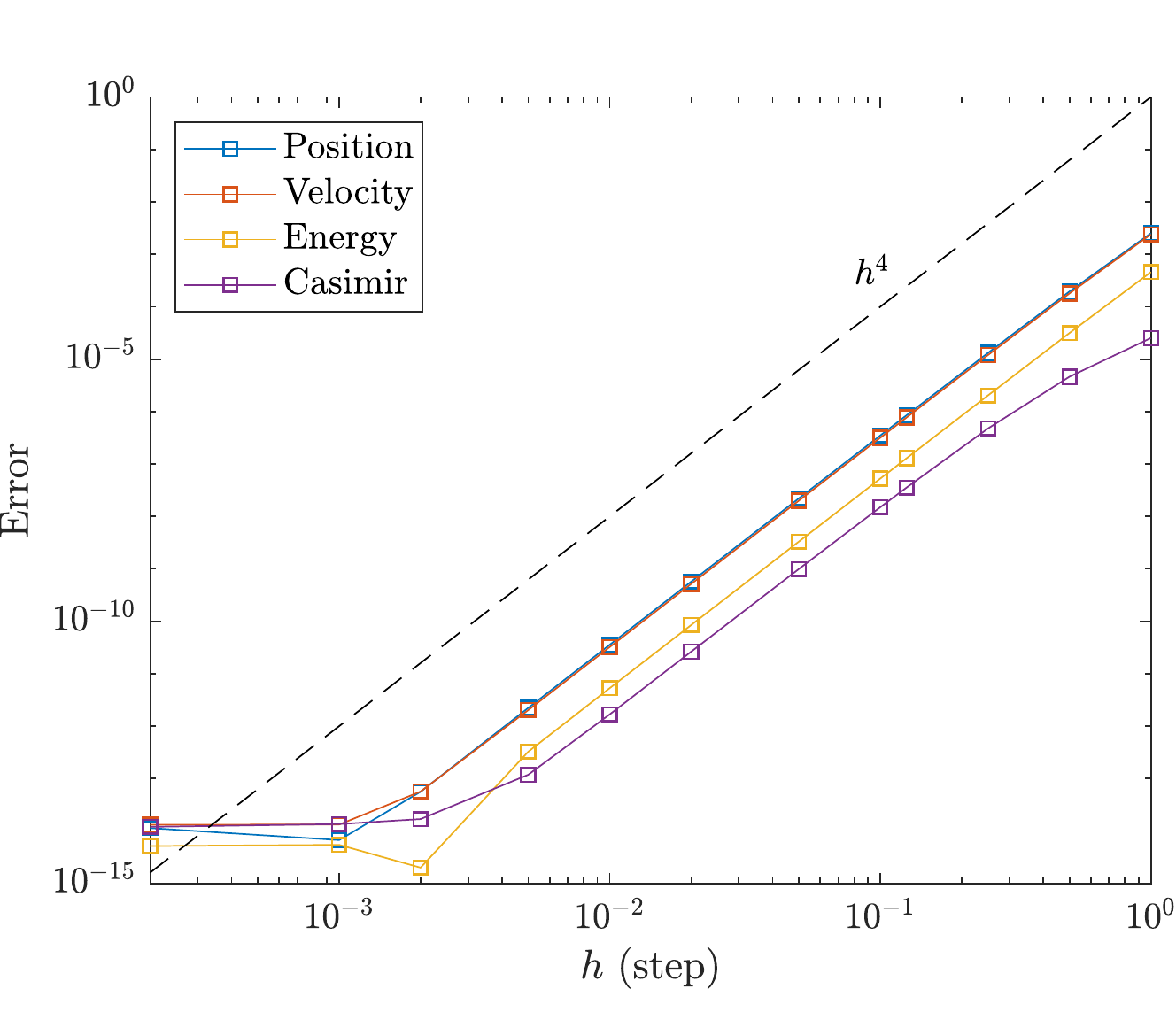}
  \end{subfigure}
  \caption{Numerical error for separate magnitudes of the model for 2-stage (left) and 3-stage (right) partitioned Lobatto methods.}
  \label{fig:LLG_order_plots}
\end{figure}

We chose to discretize the corresponding generalized Lagrangian, $\boldsymbol{\ell}_{f}$, using Lobatto schemes of $2$ (trapezoidal rule) and $3$ stages only, as Lie group integrators are computationally more demanding. The order of an $s$-stage Lobatto method is $p = 2 s - 2$, so the resulting numerical methods are of order $2$, $4$ respectively. As retraction we have used the standard Cayley map, $\mathrm{cay}$. The parameters used for the numerical simulations shown here are $I = \mathrm{diag}(I_x, I_y, I_z) = (1/2, 2, 1)$ and $\alpha = 1$, for no particular reason. The other choices of parameters that were tested showed essentially the same behaviour. We run each simulation for a total of $1$ unit of simulation time with several different choices of step-size $h$ ranging between $1 \cdot 10^{-4}$ and $1$ and measure numerical error as the difference between the final value of the magnitude being studied found for a reference simulation and the corresponding one for the value we want to study. In this case our reference is taken as the simulation with the finest step-size. The initial values chosen for the results in figures \ref{fig:LLG_casimir_plots} and \ref{fig:LLG_order_plots} are $(\Omega_x, \Omega_y, \Omega_z) = (1/\sqrt{2},0,1/\sqrt{2})$.

For the resolution of the resulting non-linear system of equations derived for each method, we used MATLAB's \verb|fsolve| with \verb|TolX=1e-12| and \verb|TolX=1e-14| respectively.

\subsubsection{Relaxed rigid body}
For our second set of numerical tests we have chosen the \emph{so-called} relaxed rigid body \cite{Morrison86}, which shares the same Lagrangian as our former example, but with different forcing. This time the force is of the form
\begin{equation*}
f = \beta (\mathbf{\Omega} \times \mathbf{M}) \times \mathbf{\Omega}
\end{equation*}
with $\beta \in \mathbb{R}$ a constant. Therefore, the equations of motion are
\begin{equation*}
\dot{\mathbf{M}} = \mathbf{M} \times \mathbf{\Omega} + \beta (\mathbf{\Omega} \times \mathbf{M}) \times \mathbf{\Omega}\,.
\end{equation*}

This is a simple model for so-called \emph{metriplectic system} (\cite{BMR2013, Morrison86,Ott1,Ott2}), which, in contrast with the former example, is known to preserve energy but not the Casimir function $C$. Both properties are a geometric expression of the First and Second Laws of Thermodynamics for a given system.  In fact, it can be shown that
\begin{equation*}
\frac{\mathrm{d} C}{\mathrm{d} t} \geq 0
\end{equation*}
where now $C$ is playing the role of the entropy function. 
As usual, variational integrators  with forcing do not exactly preserve energy or the Casimir exactly, and therefore, the results we show here are not surprising.

\begin{figure}[!ht]
  \centering
  \begin{subfigure}[b]{0.45\textwidth}
  	\includegraphics[scale=0.50, clip=true, trim=33mm 90mm 44mm 90mm]{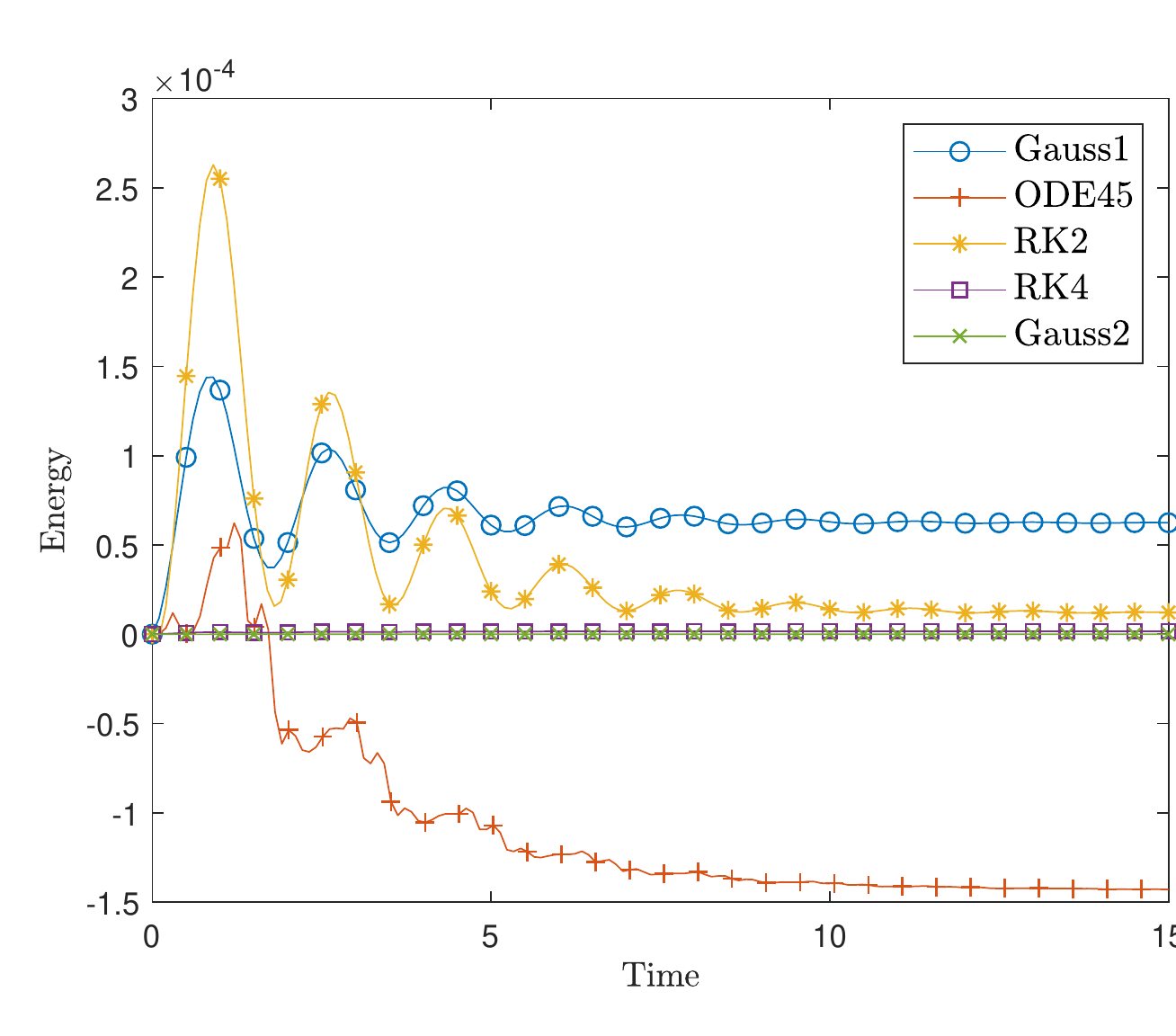}
  \end{subfigure}
  \hspace{0.5cm}
  \begin{subfigure}[b]{0.45\textwidth}
  	\includegraphics[scale=0.50, clip=true, trim=33mm 90mm 42mm 90mm]{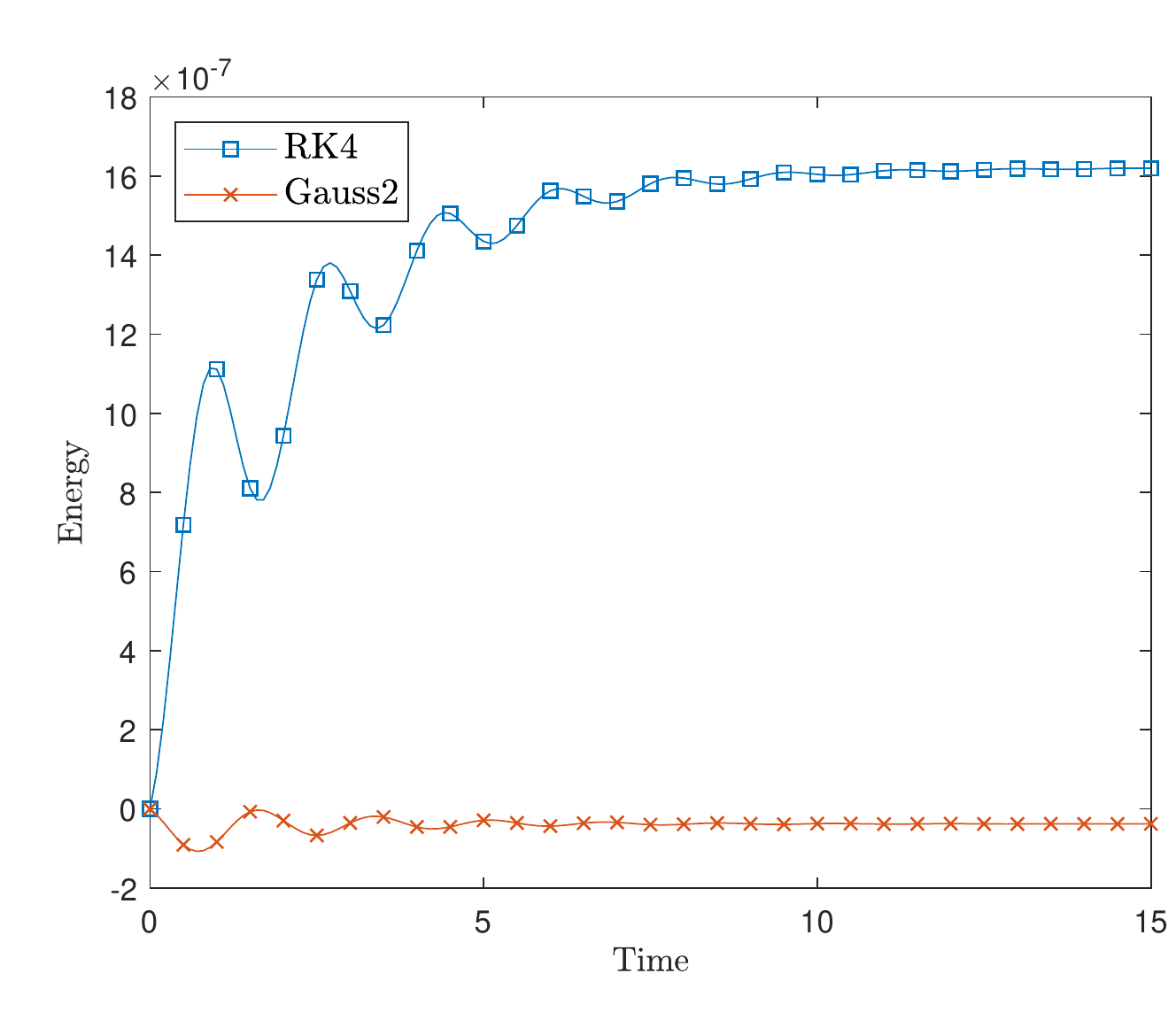}
  \end{subfigure}
  \caption{Energy. The difference between the initial value and its value at any other moment is displayed. On the left we have a battery of different integrators of different orders. On the right we have two constant-step order 4 integrators, the classic Runge-Kutta 4 and a 2-stage partitioned Gauss method. All constant-step methods use $h = 0.1$. As it is well known, the energy behaviour of ODE45 is rather poor.}
  \label{fig:relaxed_energy_plots}
\end{figure}

\begin{figure}[!ht]
  \centering
  \begin{subfigure}[b]{0.45\textwidth}
  	\includegraphics[scale=0.50, clip=true, trim=33mm 90mm 44mm 90mm]{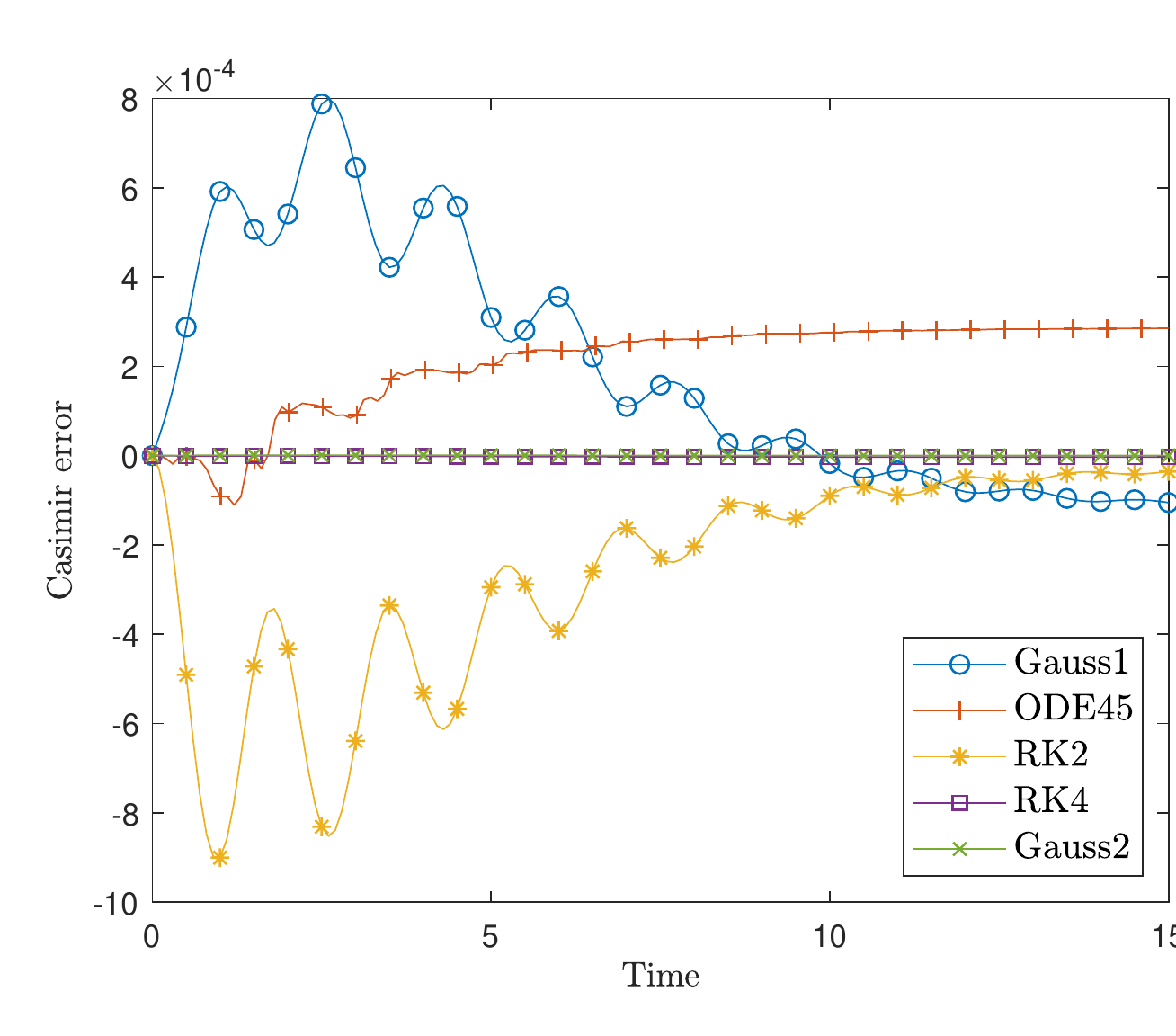}
  \end{subfigure}
  \hspace{0.5cm}
  \begin{subfigure}[b]{0.45\textwidth}
  	\includegraphics[scale=0.50, clip=true, trim=33mm 90mm 42mm 90mm]{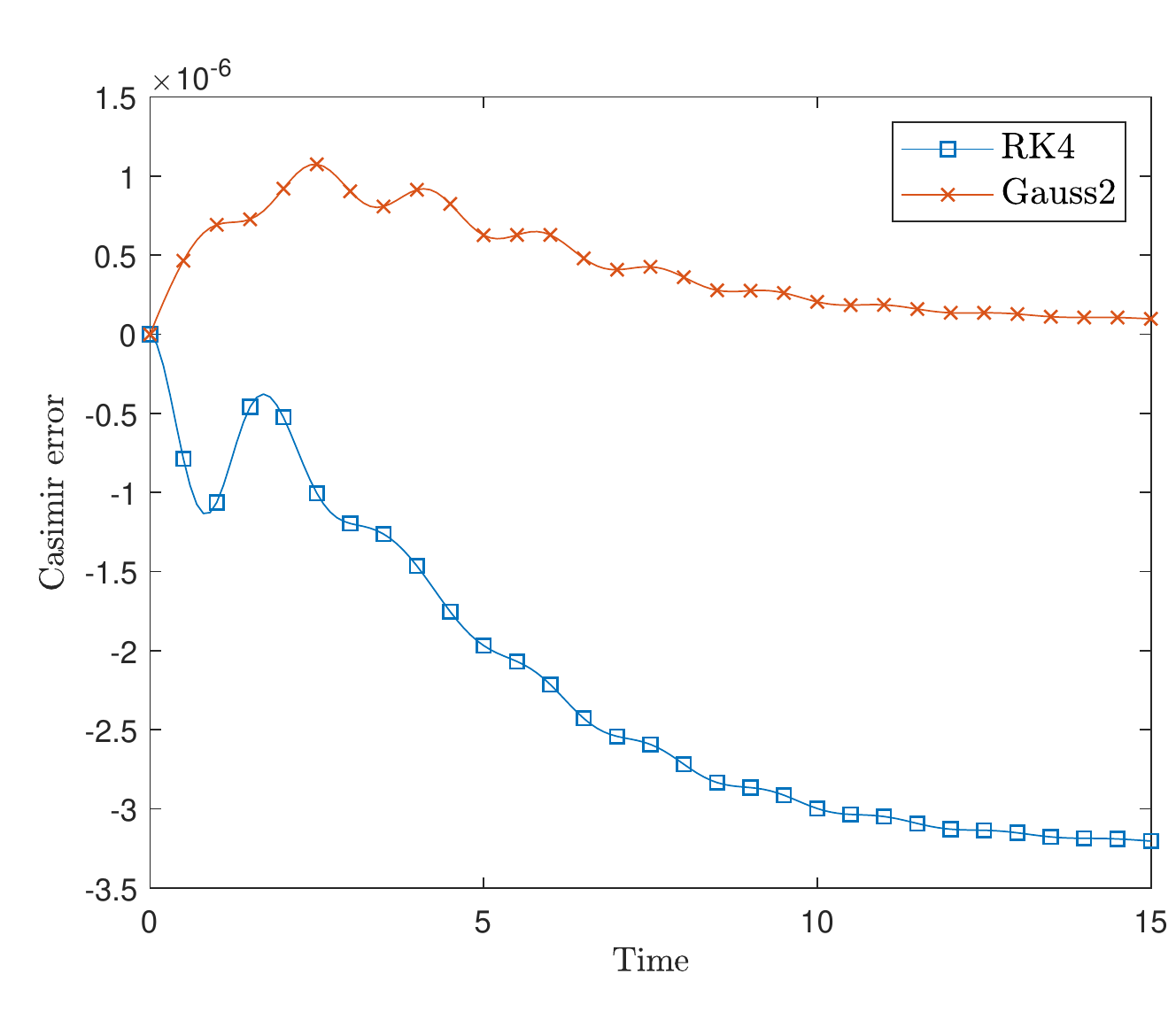}
  \end{subfigure}
  \caption{Casimir. We display the difference between the computed value and a reference value (computed with a 3-stage partitioned Gauss method and $h = 0.01$).}
  \label{fig:relaxed_casimir_plots}
\end{figure}

\begin{figure}[!ht]
  \centering
  \begin{subfigure}[b]{0.45\textwidth}
  	\includegraphics[scale=0.50, clip=true, trim=33mm 90mm 44mm 90mm]{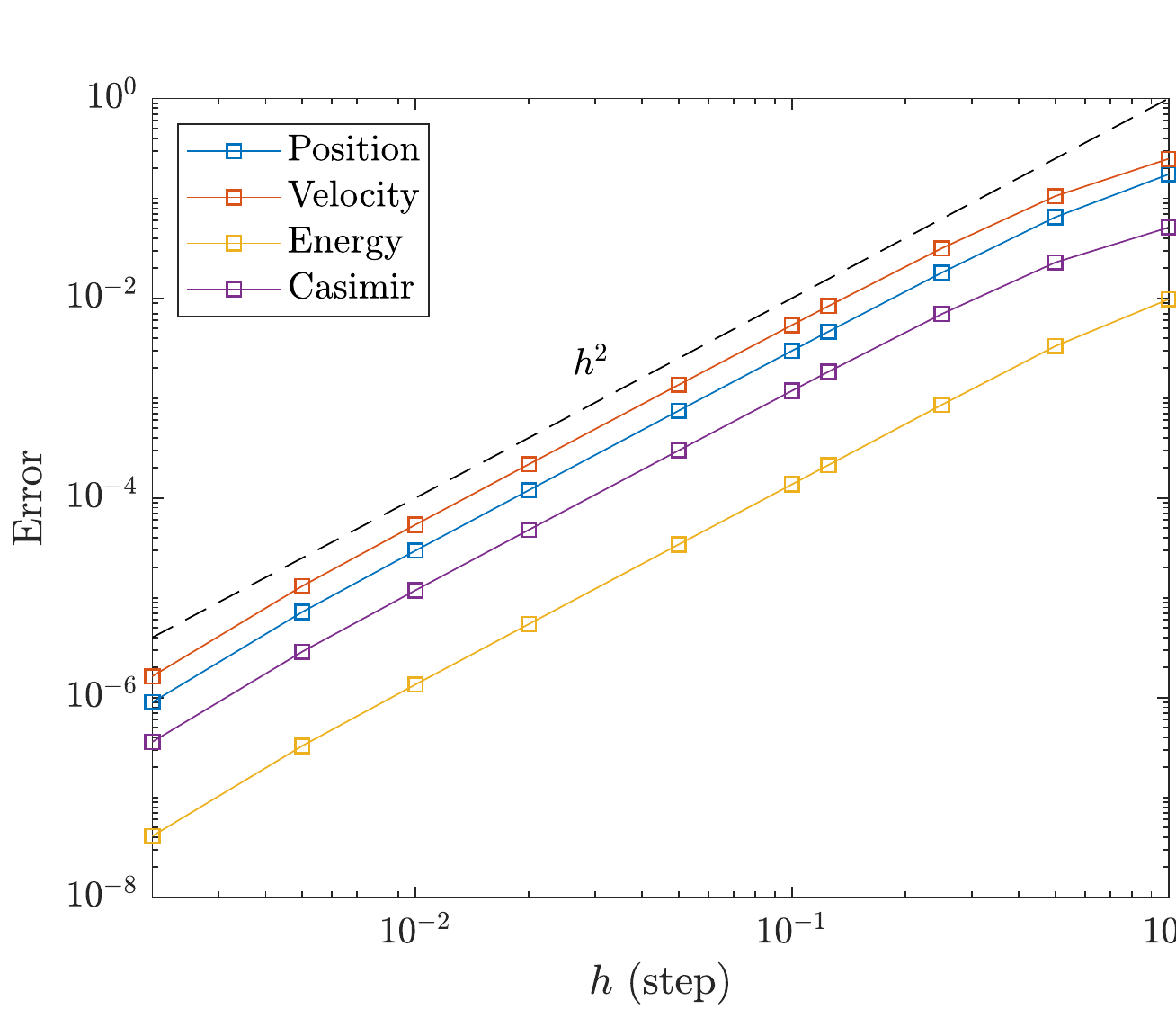}
  \end{subfigure}
  \hspace{0.5cm}
  \begin{subfigure}[b]{0.45\textwidth}
  	\includegraphics[scale=0.50, clip=true, trim=33mm 90mm 42mm 90mm]{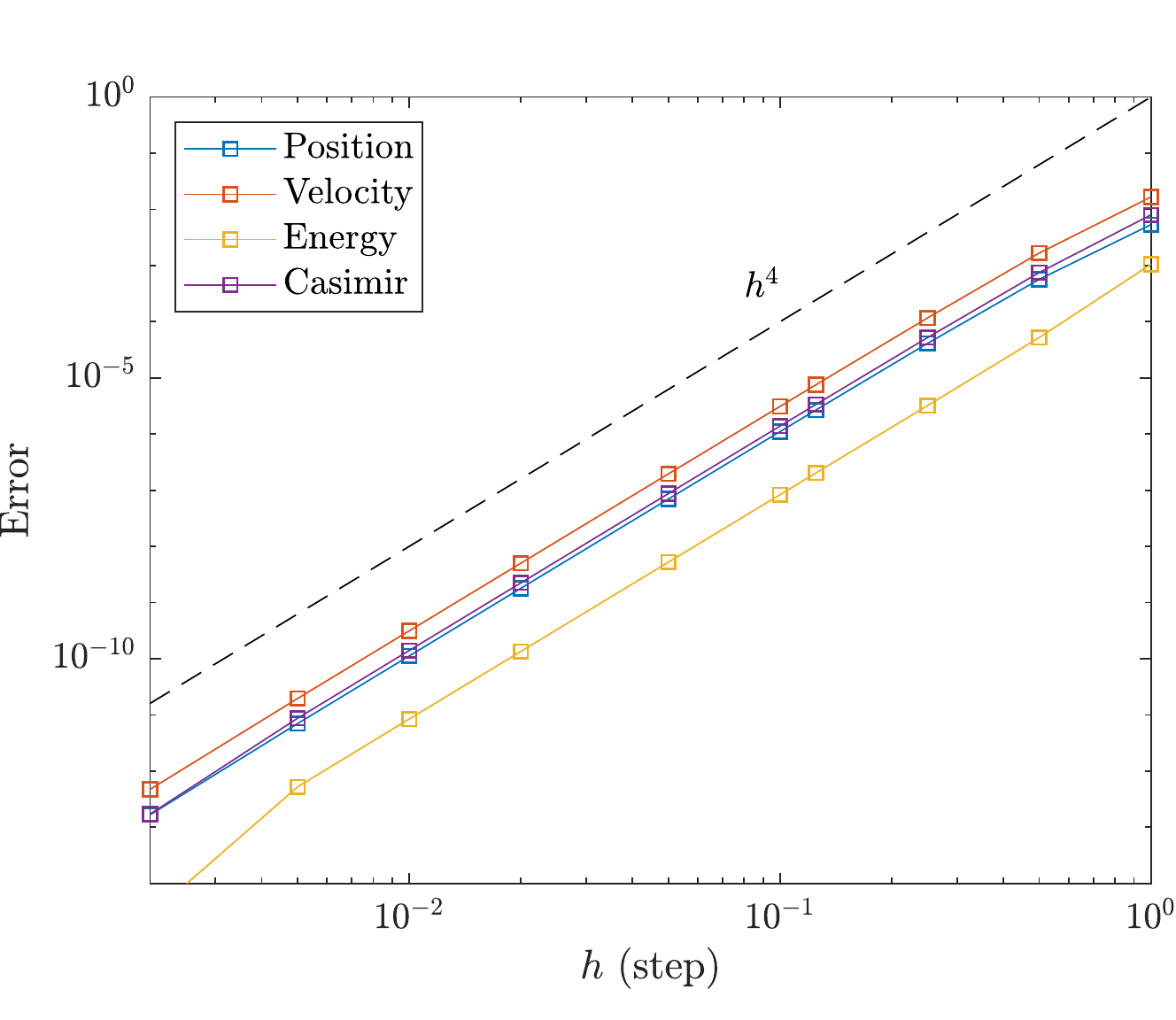}
  \end{subfigure}
  \par
  \begin{subfigure}[b]{0.45\textwidth}
  	\includegraphics[scale=0.50, clip=true, trim=33mm 90mm 42mm 90mm]{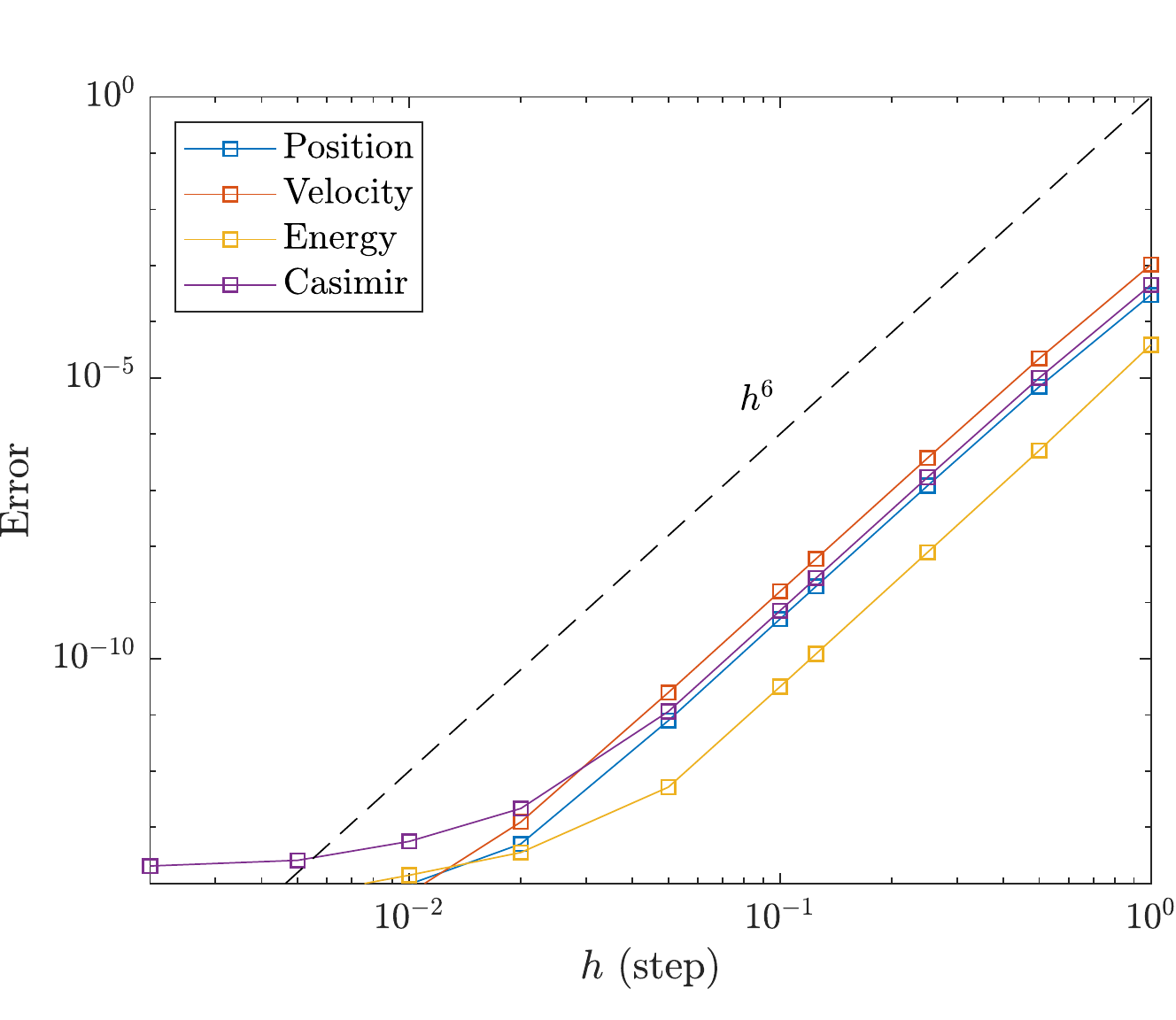}
  \end{subfigure}
  \hspace{0.5cm}
  \raisebox{1.5cm}{\begin{minipage}[b]{0.45\textwidth}
    \caption{Numerical error for separate magnitudes of the model for 1-stage (top left), 2-stage (top right) and 3-stage (bottom left) partitioned Gauss methods.}
  \end{minipage}}
  \label{fig:relaxed_order_plots}
\end{figure}

This time we chose to discretize the corresponding generalized Lagrangian, $\boldsymbol{\ell}_{f}$, using Gauss schemes of $1$ (midpoint), $2$ and $3$ stages. The order of an $s$-stage Gauss method is $p = 2 s$, so the resulting numerical methods are of order $2$, $4$ and $6$ respectively. We used $\beta = 0.1$ and we kept the rest of the parameters and conditions the same as in the LLG example.

\section{Conclusions and Future work}
In this paper, we have discussed how a Lagrangian or Hamiltonian system with forcing reduced by a Lie group of symmetries can be treated as a regular Lagrangian system in a higher dimensional space. For it, we add new phase variables in a very geometric way. It is interesting to observe that the geometry involved is related to the notion of Poisson groupoid. As a main implication we deduce the variational error analysis for Euler-Poincar\'e equations with forcing. 

In a future paper, we will study the general case of Euler-Poincar\'e equations subjected to double bracket dissipation \cite{BKMR} as in subsection \ref{subsection1}. We will study concrete discretizations of the forces in such a way the system {\sl exactly} preserves the coisotropic orbits as is done by the continuous case. Moreover, we want to check if the duplication of variables technique allows us to understand rigourously why the resulting integrators show excellent energy decay rate tracking properties in many concrete examples.

 \section*{Acknowledgements}
 
 The authors have been partially supported by Ministerio de Econom{\'\i}a, Industria y Competitividad (MINEICO, Spain) under grants MTM 2013-42870-P, MTM 2015-64166-C2-2P, MTM2016-76702-P and ``Severo Ochoa Programme for Centres of Ex- cellence" in R\&D (SEV-2015-0554).
 
%\bibliography{References}
\printbibliography

\end{document}